\documentclass[10pt,journal,comsoc]{IEEEtran}
\usepackage{graphicx} 
\usepackage{graphics}
\usepackage{multirow}
\usepackage{textcomp}
\usepackage{array}
\usepackage{xcolor}
\usepackage[hidelinks]{hyperref}

\usepackage[caption=false,font=footnotesize]{subfig}

\usepackage{tabularx}
\newcolumntype{Y}{>{\centering\arraybackslash}X} 

\usepackage{tabularray}
\usepackage{eso-pic}

\usepackage{tikz}
\usetikzlibrary{shapes, arrows.meta, positioning, fit, backgrounds, calc, shadows}

\usepackage{amsmath,amssymb,amsfonts}
\usepackage{bm}

\usepackage[algo2e,linesnumbered,ruled]{algorithm2e}

\usepackage{amsthm}
\theoremstyle{definition}

\newtheorem{proposition}{Proposition}

\newtheoremstyle{remarkstyle}
  {3pt}                
  {3pt}                
  {}                   
  {}                   
  {\bfseries}          
  {.}                  
  {.5em}               
  {}                   

\theoremstyle{remarkstyle}
\newtheorem{remark}{Remark} 

\SetKwProg{Fn}{Function}{}{end}

\SetAlgoSkip{}

\newcolumntype{M}[1]{>{\centering\arraybackslash}m{#1}}

\DeclareSymbolFont{symbolsC}{U}{txsyc}{m}{n}
\DeclareMathSymbol{\notniFromTxfonts}{\mathrel}{symbolsC}{61}

\title{ORCHID: Fairness-Aware Orchestration in Mission-Critical Air-Ground Integrated Networks}

\author{Chuan-Chi~Lai,~\IEEEmembership{Member,~IEEE}, and 
    Chi Jai Choy
    \IEEEcompsocitemizethanks{
        \IEEEcompsocthanksitem{This research was supported by the National Science and Technology Council, Taiwan, under Grant Nos. NSTC 114-2221-E-194-062- and NSTC 115-2221-E-194-042-MY2.
		This work was also partially supported by the Advanced Institute of Manufacturing with High-tech Innovations (AIM-HI) from the Featured Areas Research Center Program within the framework of the Higher Education Sprout Project by the Ministry of Education (MOE) in Taiwan. \emph{(Corresponding author: Chuan-Chi~Lai.)}}
        \IEEEcompsocthanksitem{Chuan-Chi Lai is with the Department of Communications Engineering, National Chung Cheng University, Minxiong Township, Chiayi County 621301, Taiwan, and also with the Advanced Institute of Manufacturing with High-tech Innovations (AIM-HI), National Chung Cheng University, Minxiong Township, Chiayi County 621301, Taiwan (e-mail: chuanclai@ccu.edu.tw).}
        \IEEEcompsocthanksitem{Chi Jai Choy is with the Department of Information Engineering and Computer Science, Feng Chia University, Taichung 407102, Taiwan.}
    }
}

\IEEEtitleabstractindextext{%
\begin{abstract}
Uncrewed Aerial Vehicles (UAVs) provide on-demand wireless coverage for mission-critical Air-Ground Integrated Networks (AGINs). However, DRL-based orchestration faces challenges from multi-agent non-stationarity and the trade-off between energy efficiency (EE) and service fairness. To address these issues, we propose ORCHID (Orchestration of Resilient Coverage via Hybrid Intelligent Deployment), a stability-enhanced two-stage learning framework. First, GBS-aware topology partitioning provides a topology-informed warm start for UAV deployment. Second, a Reset-and-Finetune (R\&F) mechanism is integrated into Multi-Agent Proximal Policy Optimization (MAPPO) to support stable late-stage policy refinement through event-triggered rollback and learning-rate adaptation. To study the efficiency-fairness trade-off, ORCHID employs a Max-Min Fairness (MMF)-oriented bandwidth allocation based on a closed-form Egalitarian Bargaining Game (EBG) solution and is compared with Proportional Fairness (PF). Experiments show stronger served-user bottleneck-rate protection under MMF, together with competitive coverage and user-rate fairness and favorable energy-efficiency and UAV-load-balancing performance under zero-shot scalability and mobility tests. Overall, ORCHID provides an effective balance among user-level service protection, energy efficiency, and policy stability for resilient mission-critical UAV networks.
\end{abstract}

}

\begin{document}

\setlength{\abovedisplayskip}{5pt}
\setlength{\belowdisplayskip}{5pt}
\setlength{\abovedisplayshortskip}{5pt}
\setlength{\belowdisplayshortskip}{5pt}

\setlength{\textfloatsep}{12pt plus 2pt minus 2pt}
\setlength{\floatsep}{12pt plus 2pt minus 2pt}
\setlength{\intextsep}{12pt plus 2pt minus 2pt}

\AddToShipoutPictureBG{
    \AtPageUpperLeft{%
        \raisebox{-1.3cm}[\height][\depth]{
            \makebox[\paperwidth][c]{
                \parbox{.9\textwidth}{
                    \centering
                    \footnotesize \copyright~2026 IEEE. Personal use of this material is permitted. Permission from IEEE must be obtained for all other uses, in any current or future media, including reprinting/republishing this material for advertising or promotional purposes, creating new collective works, for resale or redistribution to servers or lists, or reuse of any copyrighted component of this work in other works.
                }%
            }%
        }%
    }%
}

\maketitle
\IEEEdisplaynontitleabstractindextext

%
\IEEEpeerreviewmaketitle

\section{Introduction}
\label{sec:intro}

\IEEEPARstart{A}{s} telecommunications evolve toward the \textit{Sixth-Generation} (6G) era, \textit{Air-Ground Integrated Networks} (AGINs) have emerged to seamlessly converge terrestrial infrastructure with \textit{Non-Terrestrial Networks} (NTN), spanning satellites to \textit{High Altitude Platforms} (HAPs) \cite{jia2025distributionally}, for ubiquitous global connectivity \cite{11047530,10579820,9275613}. Within this heterogeneous architecture, \textit{Uncrewed Aerial Vehicles} (UAVs) play a pivotal role as the flexible aerial layer. Serving as agile base stations, UAVs offer rapid wireless coverage extension and mobile edge computing capabilities \cite{zhan2025joint} to complement terrestrial networks \cite{7470933}. Unlike rigid ground infrastructure, UAVs benefit from high mobility and superior \textit{Line-of-Sight} (LoS) probabilities. These characteristics enable on-demand coverage for mission-critical scenarios, such as disaster relief operations where ground segments are compromised \cite{9177297,8642333}. In such contexts, establishing a resilient \textit{Aerial Access Network} (AAN) via collaborative orchestration is not merely a technical enhancement but a critical requirement to maintain a reliable digital lifeline for trapped or isolated users.

However, realizing collaborative AGINs in public safety scenarios faces two fundamental challenges. The first involves the perceived trade-off between system efficiency and service fairness. Prevalent schemes prioritize maximizing global metrics, specifically the aggregate system sum-rate~\cite{10122731}. While efficient for commercial networks, such strategies are often detrimental in emergency contexts. Maximizing sum-rate can bias resources toward users with strong channels and may lead to coverage starvation for tail users at the cell edge or obstructed by debris~\cite{9946428}. In mission-critical networks, such disparities can create undesirable service gaps for vulnerable users. Therefore, fairness-aware resource allocation is important not only for balancing service quality but also for reducing persistent disadvantage among users. In this work, we explicitly incorporate user-rate and load fairness through \textit{Jain's Fairness Index} (JFI)~\cite{jain1984aquantitative}.

The second challenge stems from the training instability of \textit{Multi-Agent Reinforcement Learning} (MARL). While \textit{Deep Reinforcement Learning} (DRL) has been widely applied to UAV trajectory design~\cite{drones8060214,lowe2017multi}, decentralized multi-agent coordination remains challenging. Recent studies, such as the MUL-VR framework~\cite{tang2025mulvr}, have demonstrated distributed coordination for multi-UAV collaborative tasks. Under the \textit{Centralized Training with Decentralized Execution} (CTDE) paradigm, algorithms such as \textit{Multi-Agent Deep Deterministic Policy Gradient} (MADDPG) and \textit{Multi-Agent Proximal Policy Optimization} (MAPPO) have also been adopted for multi-agent orchestration. However, simultaneous policy updates create a non-stationary learning environment and may lead to late-stage policy degradation. Deep RL can additionally exhibit primacy bias, whereby early training experience has a persistent influence on subsequent learning~\cite{nikishin2022primacy}. Such behavior is undesirable in mission-critical operations because a previously strong coordination policy may deteriorate during continued training.

A separate challenge arises from uninformed spatial initialization, which can increase early exploration overhead in a continuous UAV deployment space. Motivated by the general benefits of warm-start initialization in learning systems~\cite{ash2020warm}, we use topology-informed UAV deployment to provide a more structured starting configuration for subsequent MARL optimization.

To address these gaps, we propose ORCHID (Orchestration of Resilient Coverage via Hybrid Intelligent Deployment), a stability-enhanced framework for mission-critical AGINs. First, ORCHID integrates \textit{Ground Base Station} (GBS)-aware topology partitioning using K-Means++ clustering followed by GBS-aware filtering to mitigate the spatial cold-start problem, providing a low-complexity, topology-informed warm start for UAV deployment based on user spatial structure~\cite{SODA07_Arthur}. Second, to mitigate late-stage policy degradation, we develop a \textit{Reset-and-Finetune} (R\&F) mechanism. By conditionally triggering optimizer-state resets and learning-rate scaling when a sustained multi-metric plateau is detected, R\&F enables conservative late-stage fine-tuning and improves training stability. Moreover, distinct from efficiency-first approaches, we formulate a fairness-aware objective incorporating JFI to balance user-rate and UAV-load fairness, complemented by an explicit coverage term for reducing service holes.

The main contributions are summarized as follows:
\begin{itemize}
    \item \textbf{Two-Stage Stability-Enhanced Orchestration:} We propose ORCHID, which combines GBS-aware topology initialization with MAPPO-based R\&F. Phase I provides a topology-informed warm start, while R\&F uses event-triggered checkpoint rollback, optimizer-state reset, and learning-rate scaling to support stable late-stage policy refinement and reduce the influence of stale optimizer statistics.
    \item \textbf{Analytical Foundation for Fairness-Oriented Bandwidth Allocation:} By formulating intra-cell bandwidth allocation as an \textit{Egalitarian Bargaining Game} (EBG) \cite{nash1950bargaining, kalai1977proportional}, we derive the closed-form \textit{Max-Min Fairness} (MMF) allocation and characterize how channel heterogeneity changes the bandwidth required for egalitarian rates. This result provides the analytical basis for the MMF configuration evaluated against PF.
    \item \textbf{Empirical Validation and Zero-Shot Generalization:} Simulations isolate the complementary roles of the warm start and R\&F, and show that ORCHID maintains competitive coverage and user-rate fairness while providing stronger system EE and UAV-load fairness than the MADDPG baselines across the evaluated scalability and mobility settings. The MMF-oriented configuration further provides substantially stronger served-user bottleneck-rate protection than PF while maintaining comparable fleet-load fairness. Zero-shot tests evaluate these properties under changes in user population, swarm size, and user mobility without retraining.
\end{itemize}

The remainder of this paper is organized as follows.
Section~\ref{sec:related} reviews related work;
Sections~\ref{sec:system_model} and~\ref{sec:orchid} present the system model and ORCHID framework;
Sections~\ref{sec:theoretical_analysis} and~\ref{sec:evaluation} provide the analytical foundation and performance evaluation,
followed by Section~\ref{sec:conclusion}.

\section{Related Work}
\label{sec:related}

This section reviews related studies on UAV orchestration from optimization-based and learning-based perspectives, followed by fairness-aware resource allocation to motivate the design choices of ORCHID.

\subsection{Optimization-based Approaches}

The joint optimization of UAV placement and trajectory design has been extensively investigated to maximize wireless coverage and network performance. Early approaches were predominantly predicated on geometric models and convex optimization techniques. For instance, circle packing theory was utilized to determine optimal 3D coordinates for UAVs, with the objective of maximizing total coverage area while minimizing transmit power~\cite{7486987,R2025102621,11153951}. Similarly, heuristic strategies, such as edge-prior placement algorithms, were proposed to explicitly enhance connectivity for cell-edge users~\cite{8885992}. Nevertheless, these static deployment strategies often assume simplified channel models and uniform user distributions, which renders them ill-suited for complex urban environments or dynamic disaster scenarios.

To accommodate dynamic constraints, the research paradigm has shifted toward sequential decision-making using iterative algorithms. Methods such as \textit{Successive Convex Approximation} (SCA) and \textit{Block Coordinate Descent} (BCD) have been widely employed to tackle non-convex trajectory optimization problems~\cite{8642333, 8618602}. Recent studies have further integrated emerging technologies, including \textit{Non-Orthogonal Multiple Access} (NOMA) and \textit{Reconfigurable Intelligent Surfaces} (RIS), into UAV networks. For example, joint optimization frameworks have been formulated to enhance spectral efficiency in NOMA-assisted and STAR-RIS-assisted scenarios~\cite{10320337, li2020uav}. However, these optimization-based approaches often rely on \textit{Alternating Optimization} (AO) techniques, which can incur substantial computational complexity and typically require accurate \textit{Channel State Information} (CSI). These requirements can limit their applicability to delay-sensitive mission-critical scenarios that require frequent online adaptation.

\subsection{Learning-based Approaches}

To circumvent the computational latency and CSI dependencies of traditional optimization, DRL has gained traction as a robust alternative. Pioneering works have successfully applied DRL to UAV navigation and energy efficiency optimization~\cite{liu2020energy, hu2020reinforcement}. However, managing the dimensionality curse in complex multi-UAV scenarios remains a formidable challenge due to the exponential growth of the joint state-action space. Consequently, MARL has become the preferred methodology.

To facilitate cooperative behaviors, the CTDE framework has been established as the standard architecture. Algorithms such as MADDPG and MAPPO enable agents to share global information during training while executing independent decisions during inference~\cite{lowe2017multi, NIPS2022_1787_Yu}. Recent top-tier studies have successfully applied these frameworks to complex UAV orchestration tasks. For instance, multi-agent DRL has been utilized for distributed trajectory optimization in differentiated service scenarios~\cite{ning2024multi}, and MAPPO-based approaches have demonstrated effectiveness for hierarchical resource allocation in aerial computing systems~\cite{kang2023mappo}. Furthermore, game-theoretic DRL mechanisms have been developed to jointly optimize power allocation and 3D deployment for throughput maximization~\cite{10146333}. 

Despite these advances, non-stationarity remains an important challenge in MARL-based orchestration because simultaneous policy updates can alter the effective learning environment of each agent. Many existing studies primarily report converged or best-performing results, while late-stage policy degradation receives comparatively less attention. ORCHID addresses this issue explicitly through a two-stage learning framework with event-triggered stabilization.

\subsection{Fairness-Aware Resource Allocation}

Navigating the trade-off between system efficiency and user fairness presents a fundamental dilemma in wireless networks. Throughput-oriented resource allocation schemes often prioritize aggregate system performance, such as the system sum-rate~\cite{10122731}. While such objectives can increase total throughput, they may favor users with favorable channel conditions and leave disadvantaged users with lower absolute rates. In public safety or emergency contexts, this opportunistic behavior can create service disparities and coverage holes, motivating fairness-aware objectives for vulnerable users. 

To address this disparity, MMF has been adopted as a fairness-oriented objective. Recent literature has explored adaptive deployment approaches that incorporate fairness metrics, such as JFI~\cite{jain1984aquantitative}, to balance traffic offloading~\cite{9475471}. While traffic balancing strategies have shown promise~\cite{9946428}, jointly coupling user-level fairness with UAV trajectory and power control remains challenging because of the resulting non-convex and state-dependent optimization. ORCHID addresses this gap by combining JFI-based fairness-aware learning, an environment-side MMF-oriented bandwidth allocation, and stability-enhanced MARL optimization.

\section{System Model and Problem Formulation}
\label{sec:system_model}

This section details the system architecture of the proposed mission-critical AGIN. As illustrated in Fig.~\ref{fig:system_model}, the network orchestrates a dynamic swarm of UAVs to collaborate with surviving terrestrial infrastructure (i.e., the macro GBS), establishing a resilient digital lifeline across disaster ruins and medical relief camps. In the following subsections, we mathematically characterize the spatial user distribution, UAV mobility kinematics, air-to-ground channel propagation, and the underlying communication protocols. Building upon these models, we finally formulate the joint trajectory design and power control problem, explicitly emphasizing a fairness-aware optimization objective.

\color{black}
Consider a downlink wireless network deployed within a target geographical area of dimensions $D \times D$. The network infrastructure is heterogeneous, comprising two distinct tiers:
\begin{itemize}
    \item \textbf{Terrestrial Tier:} A fixed macro GBS is located at the center of the service area to provide baseline terrestrial coverage. The horizontal position of the GBS is fixed at the center of the service area and denoted as $\mathbf{q}_{0} = [D/2, D/2]^T$. The GBS remains a candidate service node throughout the episode; users are not permanently assigned to the terrestrial tier. The GBS operates with a constant transmit power $P_{\rm GBS}$ to maintain a stable service baseline. Unlike the energy-constrained UAVs, the GBS is connected to the power grid; thus, its energy efficiency is not the primary optimization objective in this study.
    \item \textbf{Aerial Tier:} A fleet of $N$ rotary-wing UAVs, denoted by the set $\mathcal{U} = \{u_1, u_2, \dots, u_N\}$, is deployed to assist the GBS. These UAVs act as mobile aerial base stations that can serve users whose strongest eligible link is provided by an aerial node.
\end{itemize}

The system serves a set of $M$ \textit{Ground Users} (GUs), denoted as $\mathcal{G} = \{e_1, e_2, \dots, e_M\}$. For the quasi-static training and evaluation episodes, the horizontal position of the $m$-th GU (user $e_m$) is fixed and known within each episode, represented by $\mathbf{x}_m = [x_m, y_m]^T \in \mathbb{R}^2$. The complete 3D coordinate of user $e_m$ is given by $\mathbf{l}_m = [x_m, y_m, 0]^T$, assuming all users are located at ground level ($z=0$). In the zero-shot mobility evaluation, these positions are extended to time-varying states $\mathbf{x}_m(t)$ and $\mathbf{l}_m(t)$.

\begin{figure}[!t]
    \centering
    \includegraphics[width=\linewidth]{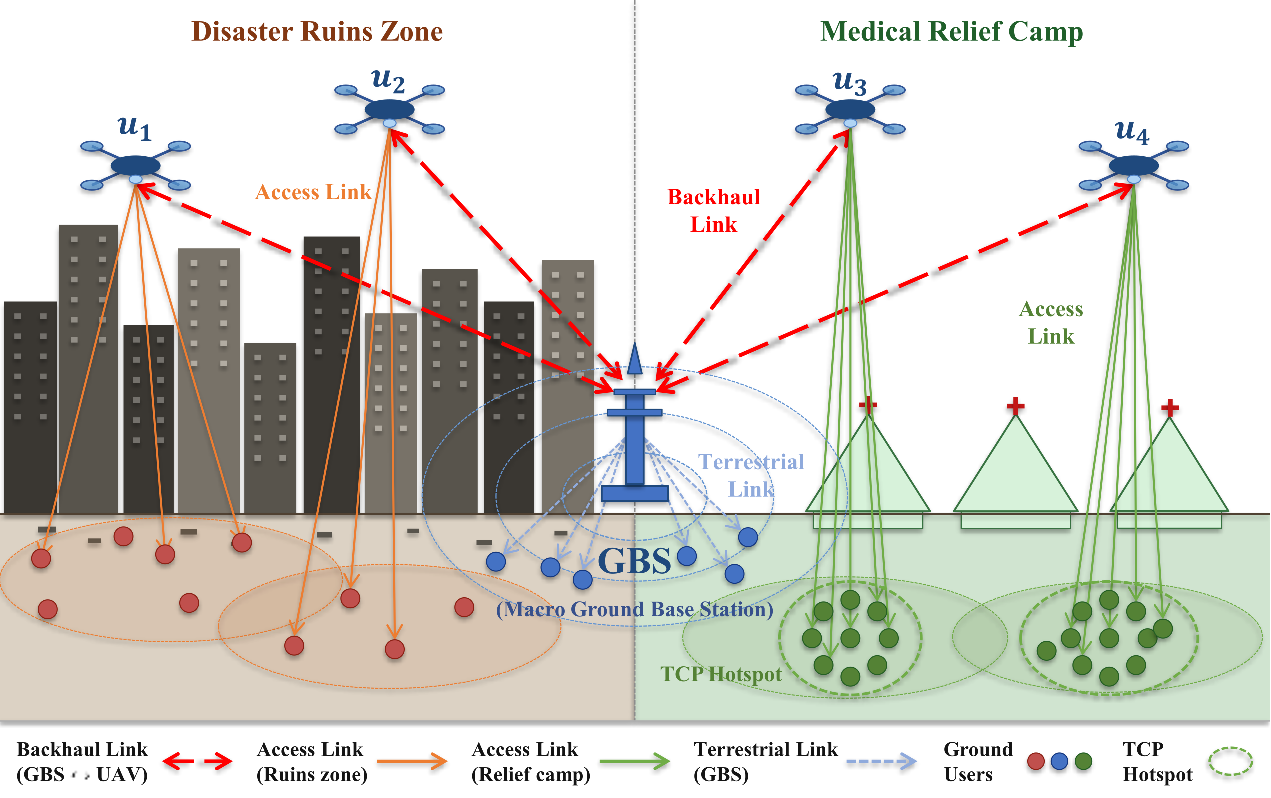}
    \caption{System model of the mission-critical Air-Ground Integrated Network (AGIN), where a swarm of UAVs assists a macro GBS to provide resilient on-demand coverage for victims across disaster ruins and medical relief camps.}
    \label{fig:system_model}
\end{figure}

\subsection{User Spatial Configuration}
\label{subsec:user_dist}

To capture heterogeneous post-disaster user deployments while remaining consistent with the implemented simulator, we consider two complementary spatial regimes.

For the clustered regime, let $\mathcal{C}=\{\mathbf c_k\}_{k=1}^{K_{\rm cl}}$ denote the set of hotspot centers, where each center is independently sampled from the uniform distribution over the square service area:
\begin{equation}
    \mathbf c_k
    \sim
    \mathcal{U}\left([0,D]^2\right),
    \qquad k=1,\ldots,K_{\rm cl}.
\end{equation}
For each hotspot center $\mathbf c_k$, the location of its $l$-th user is generated as
\begin{equation}
    \mathbf x_{k,l}
    =
    \operatorname{clip}
    \left(
        \mathbf c_k
        + \boldsymbol{\epsilon}_{k,l},
        [0,D]^2
    \right),
    \quad
    \boldsymbol{\epsilon}_{k,l}
    \sim
    \mathcal{N}(\mathbf 0,
    \sigma_{\rm scatter}^2\mathbf I_2),
\end{equation}
where $k=1,\ldots,K_{\rm cl}$ and $l=1,\ldots,M_{\rm cl}$. Hence, the clustered topology contains $M=K_{\rm cl}M_{\rm cl}$ users.

We refer to this construction as a \textit{Thomas Cluster Process (TCP)-inspired Gaussian hotspot model}, rather than a full TCP~\cite{chiu2013stochastic}, because the implementation uses uniformly sampled, fixed-number hotspot centers and a fixed number of Gaussian-scattered users per center, without the parent-point Poisson process and random daughter-count mechanism of a canonical TCP.

For the spatially dispersed regime, each user is independently and uniformly sampled over the service area:
\begin{equation}
    \mathbf x_m
    \sim
    \mathcal{U}\left([0,D]^2\right),
    \qquad m=1,\ldots,M.
\end{equation}

At the beginning of each training episode, the clustered and spatially dispersed regimes are selected with equal probability to generate the complete ground-user topology. The same two regimes are used in evaluation to assess robustness across heterogeneous spatial deployments.

\subsection{UAV Mobility Model}

The system operates in discrete time steps indexed by $t=1,2,\dots,T$, where each time slot has a duration of $\delta_t$. The mission horizon therefore consists of $T$ slots with a total duration of $T\delta_t$, reflecting the finite endurance of the rotary-wing UAVs.

Let $\mathbf{q}_n(t)=[x_n(t),y_n(t),z_n(t)]^T\in\mathbb{R}^3$ denote the instantaneous 3D position of UAV $u_n$ at time slot $t$. We define its horizontal position as $\mathbf{q}_n^{xy}(t)=[x_n(t),y_n(t)]^T$ and the admissible horizontal service region as $\mathcal{B}_{\rm xy}=[0,D]\times[0,D]$. Accordingly, the UAV position is constrained by
\begin{equation}
    \mathbf{q}_n^{xy}(t)\in\mathcal{B}_{\rm xy},
    \quad \forall n,t,
\end{equation}
and
\begin{equation}
    H_{\min}\le z_n(t)\le H_{\max},
    \quad \forall n,t.
\end{equation}

The UAV motion is modeled using discrete-time kinematics:
\begin{equation}
    \mathbf{q}_n(t+1)
    =
    \mathbf{q}_n(t)+\mathbf{v}_n(t)\delta_t,
\end{equation}
where $\mathbf{v}_n(t)$ denotes the velocity vector at time slot $t$. Its magnitude is constrained by
\begin{equation}
    \|\mathbf{v}_n(t)\|\le V_{\max},
    \quad \forall n,t.
\end{equation}

To ensure physical separation among UAVs, the Euclidean distance between any pair of UAVs must satisfy
\begin{equation}
    \|\mathbf{q}_n(t)-\mathbf{q}_j(t)\|
    \ge d_{\min},
    \quad
    \forall n\neq j,\ \forall t.
\end{equation}

\color{black}

\subsection{Channel Propagation Models}
The wireless communication environment distinguishes three link types: UAV-to-user access, GBS-to-user terrestrial access, and GBS-to-UAV downlink backhaul.

\subsubsection{UAV-to-User Link (Air-to-Ground)}
For communication between UAVs and GUs, we adopt the probabilistic \textit{Air-to-Ground} (A2G) channel model proposed in~\cite{al2014modeling}. This model captures the likelihood of LoS and \textit{Non-Line-of-Sight} (NLoS) propagation as a function of the UAV's elevation angle.

The probability of establishing an LoS link between UAV $u_n$ and user $e_m$ is given by the sigmoid function:
\begin{equation}
    \mathbf{\text{Pr}_{\text{LoS}}}(n, m, t) = \frac{1}{1 + a \exp\left(-b(\vartheta_{n,m}(t) - a)\right)},
\end{equation}
where $a$ and $b$ are environment-dependent S-curve parameters (e.g., representing suburban, urban, or dense urban clutter). The elevation angle $\vartheta_{n,m}(t)$ (in degrees) is calculated as:
\begin{equation}
    \vartheta_{n,m}(t)=\frac{180}{\pi}\arcsin\left(\frac{z_n(t)}{d_{n,m}(t)}\right),
\end{equation}
where $d_{n,m}(t) = \|\mathbf{q}_n(t) - \mathbf{l}_m\|$ represents the 3D Euclidean distance between the UAV and the user (assuming the user is at ground level $z=0$). The NLoS probability is the complement, $\mathbf{\text{Pr}_{\text{NLoS}}} = 1 - \mathbf{\text{Pr}_{\text{LoS}}}$.

The associated average path loss (in dB) is formulated as the probabilistic weighted sum of LoS and NLoS components:
\begin{equation}
    L_{n,m}^{\rm A2G}(t) = \mathbf{\text{Pr}_{\text{LoS}}} L^{\text{LoS}}_{n,m} + (1 - \mathbf{\text{Pr}_{\text{LoS}}}) L^{\text{NLoS}}_{n,m}.
\end{equation}
Here, the specific path loss for each state $\xi \in \{\text{LoS, NLoS}\}$ is defined by the \textit{free-space path loss} (FSPL) augmented by an excessive path loss factor $\eta_{\xi}$:
\begin{equation}
    L^{\xi}_{n,m} = 20 \log_{10}(d_{n,m}(t)) + 20 \log_{10}(f_c) + 20 \log_{10}\left(\frac{4\pi}{c}\right) + \eta_{\xi},
\end{equation}
where $f_c$ is the carrier frequency, $c$ is the speed of light, and $\eta_{\text{LoS}}$ and $\eta_{\text{NLoS}}$ represent the average additional attenuation due to shadowing and scattering in LoS and NLoS conditions, respectively.

\subsubsection{GBS-to-User Link (Terrestrial)}

For the terrestrial GBS-to-user link, we adopt the deterministic
path-loss component of the 3GPP urban macrocell log-distance model,
consistent with the implemented link budget~\cite{3gpp_tr36931_v17}.
Let $d_{0,m}^{\mathrm H}=\|\mathbf{x}_m-\mathbf{q}_0\|$ denote the horizontal
distance between the GBS and user $e_m$. The path loss is given by
\begin{equation}
    L_{0,m}^{\rm GBS}
    =
    128.1+
    37.6\log_{10}\!\left(
    \frac{d_{0,m}^{\mathrm H}}{1000}
    \right),
\end{equation}
where $d_{0,m}^{\mathrm H}$ is measured in meters and the division by $1000$
converts the distance to kilometers as required by the model. For
consistency with the implementation, this deterministic component is
used directly in the simulation link budget without an additional
random shadow-fading realization.

\subsubsection{GBS-to-UAV Link (Downlink Backhaul)}

For tractability, each UAV is assigned a dedicated orthogonal backhaul channel of bandwidth $B_{\text{BH}}$, so inter-UAV backhaul interference is not modeled. Each UAV receives traffic from the central GBS and forwards it to edge users via the access links.

Because both GBS and UAV antennas operate at elevated altitudes, the backhaul channel is assumed to be LoS-dominated. The backhaul \textit{Signal-to-Noise Ratio} (SNR) received at UAV $u_n$ from the GBS at time $t$ is formulated as:
\begin{equation}
    \gamma_n^{\text{BH}}(t)
    =
    \frac{P_{\text{GBS}}^{\rm W}10^{-L_{n,0}^{\text{BH}}(t)/10}}
    {N_0B_{\text{BH}}},
\end{equation}
where $P_{\text{GBS}}^{\rm W}$ is the constant GBS transmit power expressed in watts. The horizontal GBS-to-UAV distance is $d_{n,0}^{\mathrm H}(t)=\|\mathbf{q}_n^{xy}(t)-\mathbf{q}_0\|$. For the GBS-to-UAV backhaul, we adopt the deterministic LOS path-loss component of the widely used 28-GHz mmWave propagation model~\cite{akdeniz2014millimeter}:
\begin{equation}
    L_{n,0}^{\mathrm{BH}}(t)
    =
    61.4+
    20\log_{10}\!\left(d_{n,0}^{\mathrm H}(t)\right).
\end{equation}
The implementation uses this deterministic path-loss component without an additional shadow-fading realization. This 28-GHz backhaul model is used as a tractable deterministic LoS abstraction for the considered simulation scope, rather than a comprehensive mmWave link-budget model.

Based on this channel condition, the nominal backhaul transmission capacity of UAV $u_n$ is given by:
\begin{equation}
    C_n^{\text{BH}}(t)
    =
    B_{\text{BH}}
    \log_2\left(1+\gamma_n^{\text{BH}}(t)\right).
\end{equation}
We define the instantaneous backhaul demand ratio as
\begin{equation}
    \rho_n^{\text{BH}}(t)
    =
    \frac{\sum_{e_m\in\mathcal G_n(t)} R_m(t)}
    {C_n^{\text{BH}}(t)}.
\end{equation}
Here, $\sum_{e_m\in\mathcal G_n(t)}R_m(t)$ represents the aggregate access-side traffic demand induced by the current association and bandwidth allocation, whereas $C_n^{\text{BH}}(t)$ denotes the available wireless backhaul capacity. Thus, $\rho_n^{\text{BH}}(t)>1$ indicates excessive backhaul demand in the modeled demand-capacity relationship. In the implementation, this quantity is used as a soft congestion indicator and reward penalty rather than as a hard throughput constraint.
\color{black}

\subsection{Communication Model and User Association}
\label{subsec:comm_model}

To suppress inter-UAV co-channel interference, each UAV is assigned a
dedicated orthogonal access sub-band of bandwidth $B$, while the GBS uses a
separate terrestrial bandwidth $B_{\rm GBS}$. This per-UAV bandwidth
convention avoids implicitly changing the access bandwidth when the number
of UAVs varies; numerical settings are reported in Table~\ref{tab:sim_params}.

At each time step, every GU selects the strongest eligible service node. The implementation ranks eligible nodes by received power; under the configured equal access bandwidths $B=B_{\rm GBS}$ and common noise density, this is equivalent to ranking them by access-link SNR, as expressed below. Let $\mathcal{U}_0=\{0,1,\ldots,N\}$ denote the service-node set, where index $0$ denotes the GBS. Following the implemented A2G model, UAV $u_n$ is considered an association candidate only when its horizontal user distance lies within the nominal coverage radius
$r_n^{\rm cov}(t)=z_n(t)/\tan\vartheta_{\rm opt}$; the GBS remains an
eligible candidate for all users. Accordingly, define
$\mathcal{S}_m(t)=\{0\}\cup\{n\in\{1,\ldots,N\}:d_{n,m}^{\mathrm H}(t)\leq r_n^{\rm cov}(t)\}$, where
$d_{n,m}^{\mathrm H}(t)=\|\mathbf q_n^{xy}(t)-\mathbf x_m\|$ and
$\vartheta_{\rm opt}$ is the fixed optimal elevation angle used by the
implemented A2G model. The binary association variable $\alpha_{i,m}(t)\in\{0,1\}$ is then defined by
\begin{equation}\label{eq:assoc_rule}
    \begin{split}
    \alpha_{i,m}(t)=&\mathbb{I}\!\left(i=\arg\max_{j\in\mathcal{S}_m(t)}\gamma_{j,m}(t)\right)\\
    &\times \mathbb{I}\!\left(\max_{j\in\mathcal{S}_m(t)}\gamma_{j,m}(t)\geq\Gamma_{\rm req}\right), \quad \forall i\in\mathcal{U}_0,m,t.
    \end{split}
\end{equation}
Any tie in the maximum-SNR decision is resolved by a fixed deterministic
rule. Thus, a user whose best link is below $\Gamma_{\rm req}$ has
$\alpha_{i,m}(t)=0$ for all candidate nodes and is treated as unserved.
For UAV $u_n$, the associated-user set and instantaneous load are
\begin{equation}
    \mathcal{G}_n(t)=\{e_m\in\mathcal{G}:\alpha_{n,m}(t)=1\},
    \qquad K_n(t)=|\mathcal{G}_n(t)|.
\end{equation}
For coverage accounting, we use
\begin{equation}\label{eq:coverage_set}
    \mathcal{G}^{\rm cov}(t)=\left\{e_m\in\mathcal{G}:\max_{j\in\mathcal{S}_m(t)}\gamma_{j,m}(t)\geq\Gamma_{\rm req}\right\}.
\end{equation}
Under the orthogonal access assumptions used here, the access-link SNR is the applicable link-quality measure for this threshold test.

For link-budget and energy calculations, transmitter powers are expressed in watts. The UAV control variable $P_n(t)$ is maintained in mW, and its watt-valued form is defined as $P_n^{\rm W}(t)=P_n(t)/P_{\rm unit}$ with
$P_{\rm unit}=10^3$ for the mW-to-W conversion. The GBS transmit power is likewise converted from its engineering dBm setting to watts before substitution.

For an associated UAV user, the access-link SNR over the dedicated UAV
sub-band is
\begin{equation}
    \gamma_{n,m}(t)=\frac{P_n^{\rm W}(t)10^{-L_{n,m}^{\rm A2G}(t)/10}}{N_0B},
\end{equation}
where $P_n^{\rm W}(t)$ is the UAV transmit power in watts, $L_{n,m}^{\rm A2G}(t)$ is the
A2G path loss, and $B$ is the per-UAV access bandwidth. For the
GBS, the terrestrial link budget is used with the fixed GBS transmit power
and the separate bandwidth $B_{\rm GBS}$; its numerical setting is reported
in Table~\ref{tab:sim_params}:
\begin{equation}
    \gamma_{0,m}(t)=\frac{P_{\text{GBS}}^{\rm W}10^{-L_{0,m}^{\rm GBS}/10}}
    {N_0B_{\rm GBS}}.
\end{equation}

Within each service cell, the continuous bandwidth fraction
$\omega_{i,m}(t)\in[0,1]$ satisfies
$\sum_{e_m\in\mathcal{G}_i(t)}\omega_{i,m}(t)\leq1$, where
$i=0$ denotes the GBS and $i=n\geq1$ denotes a UAV. Let
$\mathcal{G}_i(t)=\{e_m\in\mathcal{G}:\alpha_{i,m}(t)=1\}$ and
$K_i(t)=|\mathcal{G}_i(t)|$. The corresponding bandwidth is
$B_i=B_{\rm GBS}$ for $i=0$ and $B_i=B$ otherwise. The resulting user rate is
\begin{equation}\label{eq:general_rate}
    R_m(t)=\sum_{i\in\mathcal{U}_0}\alpha_{i,m}(t)
    \omega_{i,m}(t)B_i\log_2\!\left(1+\gamma_{i,m}(t)\right).
\end{equation}
For the PF configuration, each service cell uses equal bandwidth sharing,
$\omega_{i,m}(t)=1/K_i(t)$. For the MMF configuration, the environment uses
the closed-form egalitarian allocation derived in Section~\ref{sec:theoretical_analysis}
for the corresponding service cell. For UAV cells, this is given by
\begin{equation}\label{eq:mmf_online_omega}
    \omega_{n,m}^{\rm MMF}(t)=
    \frac{[\log_2(1+\gamma_{n,m}(t))]^{-1}}
    {\sum_{e_j\in\mathcal{G}_n(t)}[\log_2(1+\gamma_{n,j}(t))]^{-1}}.
\end{equation}
The same deterministic weighting rule is applied to the GBS service
cell, with $B_0=B_{\rm GBS}$. Thus, bandwidth remains an environment-side
allocation and is not included in the MAPPO action vector; the policy only
controls UAV trajectories and transmit powers. A user with no eligible
associated node has $R_m(t)=0$.

\subsection{UAV Energy Consumption Model}

\label{subsec:energy_model}
In mission-critical operations, the endurance of rotary-wing UAVs is strictly constrained by their onboard battery capacity. The total power consumption of UAV $u_n$ at time $t$, denoted as $P_{\text{total},n}(t)$, consists of two main components: the communication transmit power $P_n^{\rm W}(t)$ and the \textit{propulsion power} $P_{\text{prop},n}(t)$ required for hovering and flying. According to the classical aerodynamic model for rotary-wing UAVs \cite{zeng2019energy}, the propulsion power is formulated as a function of the UAV velocity $\|\mathbf{v}_n(t)\|$:
\begin{equation}
\begin{split}
    P_{\text{prop},n}(t) =\;&
    P_0 \left( 1 + \frac{3 \|\mathbf{v}_n(t)\|^2}{U_{\text{tip}}^2} \right)
    + \frac{1}{2} d_{\text{fuse}} \rho s A \|\mathbf{v}_n(t)\|^3 \\
    &+ P_i \left(
    \sqrt{1 + \frac{\|\mathbf{v}_n(t)\|^4}{4v_0^4}}
    - \frac{\|\mathbf{v}_n(t)\|^2}{2v_0^2}
    \right)^{1/2},
\end{split}
\end{equation}
where $P_0$ and $P_i$ represent the blade profile power and induced power in hovering status, respectively. $U_{\text{tip}}$ is the tip speed of the rotor blade, $v_0$ is the mean rotor induced velocity in hover, $d_{\text{fuse}}$ is the fuselage drag ratio, $\rho$ is the air density, $s$ is the rotor solidity, and $A$ is the rotor disc area. Since $P_{\text{prop},n}(t)$ typically dominates the total energy budget (often in the scale of hundreds of Watts) compared with the UAV transmit power (whose engineering setting is reported in mW), incorporating the propulsion model is essential for a realistic evaluation of energy efficiency. The total instantaneous power is thus given by:
\begin{equation}
    P_{\text{total},n}(t) = P_n^{\rm W}(t) + P_{\text{prop},n}(t).
\end{equation}

\subsection{Problem Formulation: Multi-Objective Orchestration}

The primary objective is to jointly optimize UAV trajectories and transmit powers to balance coverage, energy efficiency, and service fairness. Let $U_{\rm total}(t)$ denote the composite utility:
\begin{equation}\label{eq:global_utility}
    U_{\rm total}(t)
    =
    w_1U_{\rm cov}(t)
    +w_2U_{\rm ee}(t)
    +w_3U_{\rm load}(t)
    +w_4U_{\rm rate}(t),
\end{equation}
where $U_{\rm cov}$, $U_{\rm ee}$, $U_{\rm load}$, and $U_{\rm rate}$ denote the high-level utility terms corresponding to coverage, energy efficiency, UAV-load fairness, and user-rate fairness, respectively. Equation~\eqref{eq:global_utility} provides the high-level decomposition of the design objectives; the implemented RL reward in \eqref{eq:hybrid_reward} augments these terms with coverage shaping and constraint penalties.

Let $\mathbf Q=\{\mathbf q_n(t)\}$ and $\mathbf P=\{P_n(t)\}$ denote the active decision variables, where $\mathbf q_n(t)$ and $P_n(t)$ represent the position and transmit power of UAV $u_n$ at time slot $t$, respectively. The association profile $\boldsymbol{\alpha}$ is induced deterministically by the strongest-eligible-link rule in~\eqref{eq:assoc_rule} rather than optimized independently. The joint optimization problem is therefore formulated as
\begin{subequations}
\begin{align}
    (\mathrm{P}1):\quad
    \max_{\mathbf Q,\mathbf P}\quad
    & \frac{1}{T}\sum_{t=1}^{T}U_{\rm total}(t)
    \\
    \mathrm{s.t.}\quad
    & \|\mathbf q_n(t+1)-\mathbf q_n(t)\|
      \le V_{\max}\delta_t,
    && \forall n,\ t<T,
    \label{const:vel}
    \\
    & \|\mathbf q_n(t)-\mathbf q_j(t)\|
      \ge d_{\min},
    && \forall n\neq j,\ t,
    \label{const:col}
    \\
    & \mathbf q_n^{xy}(t)\in\mathcal B_{\rm xy},
    && \forall n,\ t,
    \label{const:bound}
    \\
    & H_{\min}\le z_n(t)\le H_{\max},
    && \forall n,\ t,
    \label{const:alt}
    \\
    & P_{\min}\le P_n(t)\le P_{\max},
    && \forall n,\ t,
    \label{const:power}
    \\
    & \sum_{i\in\mathcal U_0}\alpha_{i,m}(t)\le1,
    && \forall m,\ t,
    \label{const:assoc}
    \\
    & \sum_{t=1}^{T}
      P_{\rm total,n}(t)\delta_t
      \le E_{\rm budget},
    && \forall n.
    \label{const:energy}
\end{align}
\end{subequations}

The association variables in~\eqref{const:assoc} are induced by \eqref{eq:assoc_rule} and are therefore not independent optimization variables; the constraint is included only to make the single-service-node association property explicit. The cumulative energy constraint in \eqref{const:energy} specifies a per-UAV battery-feasibility condition over the mission horizon. In the implementation, the displacement, horizontal-region, altitude, and transmit-power bounds are enforced directly through the action and motion mappings. Collision avoidance and backhaul congestion are handled through soft penalties in the reward, while the cumulative energy budget remains a problem-level feasibility condition rather than a separate episode-rejection rule.

\subsection{Problem Complexity}

Problem (P1) is a non-convex, non-linear optimization problem with state-dependent, discontinuous user association dynamics. The complexity arises from the non-smooth coupling between continuous decision variables $(\mathbf{Q}, \mathbf{P})$ and the induced discrete association states in \eqref{eq:assoc_rule}, the non-convex collision avoidance constraint \eqref{const:col}, and the conflicting nature of the multi-objective utility function. Since deriving an exact global optimum using traditional convex optimization or exhaustive search is computationally prohibitive for large-scale dynamic networks, this motivates the use of the proposed MARL-based ORCHID framework to obtain high-quality solutions efficiently.

\color{black}

\begin{figure*}[t]
\centering
\resizebox{\textwidth}{!}{
\begin{tikzpicture}[
    node distance=1.6cm and 1.3cm,
    >=Stealth,
    font=\sffamily\normalsize,
    block/.style={draw, rectangle, rounded corners, minimum width=2.6cm, minimum height=1.3cm, align=center, fill=white, line width=1pt},
    process/.style={block, fill=blue!10, draw=blue!60, very thick},
    marl/.style={block, fill=orange!10, draw=orange!60, very thick, minimum width=3.4cm, minimum height=1.6cm},
    decision/.style={draw, diamond, aspect=1.8, fill=yellow!10, draw=yellow!60, very thick, align=center, inner sep=2pt, minimum width=2.2cm},
    mechanism/.style={block, fill=red!5, draw=red!60, dashed, very thick, minimum width=3.8cm},
    input/.style={trapezium, trapezium left angle=70, trapezium right angle=110, draw, fill=gray!10, very thick, minimum width=2.4cm, align=center},
    group/.style={draw, dashed, inner sep=0.6cm, rounded corners, fill opacity=0.3, line width=0.8pt}
]

    \node[input] (users) {Coarse User Spatial\\Information\\$\mathbf{X}$};

    \node[process, right=1.6cm of users] (kmeans) {\textbf{GBS-Aware}\\\textbf{Heterogeneous Clustering}};

    \node[block, below=1.5cm of kmeans, text width=3cm] (centroids) {Centroids $\boldsymbol{\mu}_n$\\$\downarrow$\\Initial Pos $\mathbf{q}(0)$};
    
    \node[marl, right=3.5cm of kmeans] (actor) {\textbf{Fairness-Aware MARL Optimization}\\Centralized Critic \& Distributed Actors};
    
    \node[block, above=1.5cm of actor, fill=green!5, draw=green!60] (env) {Environment\\(Kinematics \& Channel)};
    
    \node[block, below=1.5cm of actor, fill=purple!5, draw=purple!60, text width=4.5cm] (reward) {Multi-Objective Reward\\{\footnotesize Coverage, EE, Load, Rate}\\{\footnotesize Constraint Penalties \& Shaping}};

    \node[decision, right=1.8cm of actor] (check) {R\&F\\Trigger\\Check};
    
    \node[mechanism, above=1cm of env, text width=4.25cm] (reset) {\textbf{Reset-and-Finetune}\\Rollback Best Checkpoint\\Reset Optimizer ($\mathbf{m}, \mathbf{v}$)\\LR Scaling ($\kappa$)\\Reset $\log\sigma$};

    \node[input, right=2.2cm of check] (output) {Trained\\Policy $\pi_\theta$};

    \coordinate (limit_left) at ([xshift=-2.5cm]env.west);
    \coordinate (limit_right) at ([xshift=0.15cm]check.east);
    \coordinate (limit_top1) at ([yshift=0.2cm]kmeans.north);
    \coordinate (limit_top2) at ([yshift=0.2cm]reset.north);

    
    \draw[->, very thick] (users) -- (kmeans);
    \draw[->, very thick] (kmeans) -- (centroids);
    \draw[->, very thick] (centroids.east) -- ++(1.75,0) |- ([yshift=-0.4cm]actor.west);

    \draw[->, very thick] (actor.north) -- node[right, font=\small] {Action $\mathbf{a}_n(t)$} (env.south);
    \draw[->, very thick] (env.east) -- ++(1.7,0) |- node[pos=0.25, right, font=\small, yshift=2pt] {Obs. $\mathbf{o}_n(t)$} ([yshift=0.5cm]actor.east);
    
    \draw[->, very thick] (env.west) -- ++(-1.8,0) |- (reward.west);
    \draw[->, very thick] (reward.north) -- node[right, font=\small] {Reward $r_n$} (actor.south);

    \draw[->, very thick] (actor.east) -- (check.west);
    \draw[->, very thick] (check.east) -- node[above, font=\small] {No Trigger} (output.west);
    
    \draw[->, red, line width=2pt, rounded corners=15pt] (check.north) -- node[right, color=red, font=\small] {Trigger R\&F} ++(0, 3) |- (reset.east);
    \draw[->, red, line width=2pt, rounded corners=15pt] (reset.west) -| ([xshift=-2.5cm]env.west) |- ([yshift=0.4cm]actor.west);

    \begin{pgfonlayer}{background}
        \node[fit=(kmeans)(centroids)(limit_top1), group, fill=blue!5, label={[blue!80, font=\large\bfseries, anchor=north, yshift=-0.1cm]north:{Phase I: Initialization}}] {};
        
        \node[fit=(actor)(env)(reward)(reset)(check)(limit_left)(limit_top2)(limit_right), group, fill=orange!5, label={[orange!80, font=\large\bfseries, anchor=north, yshift=-0.1cm]north:{Phase II: Resilient Fine-Tuning}}] {};
    \end{pgfonlayer}
\end{tikzpicture}
}
\caption{Overall architecture of the proposed ORCHID framework, consisting of GBS-aware heterogeneous initialization (Phase I) and fairness-aware MAPPO optimization with R\&F-based stability enhancement (Phase II). User topologies are regenerated episode-wise under the clustered and spatially dispersed sampling scheme.}
\label{fig:framework}
\end{figure*}

\section{The Proposed ORCHID Framework}
\label{sec:orchid}

In this section, we propose \textit{Orchestration of Resilient Coverage via Hybrid Intelligent Deployment (ORCHID)}, a stability-enhanced two-stage learning framework. As illustrated in Fig.~\ref{fig:framework}, the proposed architecture synergizes unsupervised topology learning with on-policy MARL optimization. This hybrid design addresses early exploration overhead through topology-aware initialization and supports stable late-stage policy refinement while balancing coverage, energy efficiency, and service fairness.

\subsection{Framework Overview}
The orchestration process operates in two sequential phases to decouple global partitioning from fine-grained local optimization:
\begin{itemize}
    \item \textbf{Phase I (Initialization):} This phase reduces the exploration overhead of MARL through GBS-aware heterogeneous clustering. The user locations together with the fixed GBS point are clustered into $N+1$ groups; the centroid nearest to the GBS is excluded from UAV initialization, and the remaining $N$ centroids provide the topology-aware warm start.
    
    \item \textbf{Phase II (Resilient Fine-Tuning):} This phase executes a fairness-aware MARL optimization process using MAPPO to jointly optimize UAV trajectories and transmit powers. To address multi-agent non-stationarity and late-stage training instability, R\&F evaluates recent fairness and reward trends every $W$ episodes and, when the trigger condition is satisfied, rolls back the policy, resets the optimizer state, and resumes fine-tuning with a reduced learning rate.
\end{itemize}

\subsection{Phase I: GBS-Aware Heterogeneous Initialization}

Conventional MARL-based UAV control often begins from uninformed or random spatial deployments. In continuous coverage control, such initial configurations can induce unnecessary early exploration and inefficient transient trajectories, resulting in a spatial cold-start effect. To mitigate this issue, ORCHID leverages the underlying user spatial configuration to construct a GBS-aware topology-informed initialization.

We formulate the initial deployment problem as a clustering optimization task. Let $\mathbf{X}=[\mathbf{x}_1,\ldots,\mathbf{x}_M]\in\mathbb{R}^{2\times M}$ collect the horizontal coordinates of the GUs $\mathcal{G}=\{e_1,\ldots,e_M\}$.
Phase I assumes that coarse horizontal user locations are available from a preliminary sensing or estimation layer, for example through pilot-based AoA/ranging or ISAC-assisted sensing \cite{liu2022integrated, cui2021integrating}. In the simulations, this layer is abstracted by directly providing the generated user coordinates $\mathbf X$; sensing errors are outside the present scope. This Phase-I spatial input is used only to construct the initialization and is not included in the decentralized actor observation during Phase II.

To account for the heterogeneous deployment of UAVs and the GBS, we apply K-Means++ initialization followed by K-Means updates to the augmented horizontal point set $\mathcal X^{+}=\{\mathbf x_1,\ldots,\mathbf x_M,\mathbf q_0\}$ with $K=N+1$ clusters, consistent with the implementation. The clustering objective minimizes the \textit{Within-Cluster Sum of Squares} (WCSS):
\begin{equation}
    \text{WCSS} = \sum_{k=1}^{N+1} \sum_{\mathbf y \in C_k} \|\mathbf y - \boldsymbol{\mu}_k\|^2,
\end{equation}
where $C_k \subset \mathcal X^{+}$ and $\boldsymbol{\mu}_k \in \mathbb{R}^2$ denote the $k$-th cluster and its centroid, respectively.

Next, we apply \textit{GBS Filtering}. We calculate the Euclidean distance between each cluster centroid $\boldsymbol{\mu}_k$ and the fixed GBS location $\mathbf{q}_0$ and exclude the centroid nearest to the GBS from the UAV initialization set. This filtering step is used only to avoid redundant UAV initialization near the GBS; it does not permanently assign the corresponding users to the terrestrial tier, since user association during operation is determined dynamically by the strongest eligible link.

Finally, we perform \textit{UAV Assignment}. The remaining $N$ centroids determine the initial horizontal coordinates of the $N$ UAVs. Formally, let $\mathcal{M}_{\rm UAV}$ be the set of the $N$ selected centroids. The initial 3D position of the $n$-th UAV is set as:
\begin{equation}
    \mathbf{q}_n(0) = [\boldsymbol{\mu}_n^T, H_{\text{init}}]^T, \quad \forall \boldsymbol{\mu}_n \in \mathcal{M}_{\rm UAV},
\end{equation}
where $H_{\text{init}}$ is a pre-configured cruising altitude. This strategy partitions the spatial demand into $N+1$ service-oriented regions and assigns the resulting UAV-side centroids to the aerial fleet while avoiding redundant initialization near the GBS.

\subsection{Phase II: Fairness-Aware MARL Optimization}
\label{subsec:phase2}
Following the coarse initialization, the problem of fine-grained trajectory design and power control is formulated as a \textit{Decentralized Partially Observable Markov Decision Process} (DEC-POMDP). We employ the MAPPO algorithm, which operates under the CTDE paradigm.

\subsubsection{Observation and State Spaces}

During decentralized execution, each UAV agent $n$ observes a local 28-dimensional vector constructed from its operating status and local network information:
\begin{equation}
\begin{split}
    \mathbf o_n(t)=\Big[&\tilde{\mathbf q}_n(t),\tilde P_n(t),\tilde d_{n,\min}(t),
    \mathbf z_n^{\rm aux}(t),\\
    &\tilde d_{n,0}(t),\tilde K_n(t),\tilde\gamma_n^{\rm BH}(t)\Big]^T.
\end{split}
\end{equation}
Here, $\tilde{\mathbf q}_n(t)$ is the normalized 3D UAV position, obtained by scaling the $x$-, $y$-, and $z$-coordinates by $D$, $D$, and $H_{\max}$, respectively. The normalized transmit power is $\tilde P_n(t)=[P_n(t)-P_{\min}]/[P_{\max}-P_{\min}]$, and $\tilde d_{n,\min}(t)$ is the normalized horizontal distance to the nearest ground user, obtained by dividing by $D$. The auxiliary block $\mathbf z_n^{\rm aux}(t)\in\mathbb R^{20}$ is zero-valued in the current implementation and is retained to preserve the fixed input dimension. The remaining terms are the normalized horizontal GBS distance,
\begin{equation}
    \tilde d_{n,0}(t)
    =
    \frac{\|\mathbf q_n^{xy}(t)-\mathbf q_0\|}{D/2},
\end{equation}
the local served-user count normalized by the total number of users, and the backhaul SNR represented on a logarithmic scale, respectively. Since the GBS is located at the center of the square service area, $D/2$ provides a natural reference distance. The resulting maximum value over the square service area is $\sqrt{2}$ rather than exactly one. Thus, the implemented actor input uses local geometric and link information and does not include explicit AoA measurements or the global user topology.

During centralized training, the implemented multi-head critic receives the global state $s(t)$ together with the joint action vector. For clarity, let $\mathbf{Q}(t)=[\mathbf{q}_1^T(t),\ldots,\mathbf{q}_N^T(t)]^T$ denote the instantaneous stacked UAV-position vector, and let $\mathbf{P}(t)=[P_1(t),\ldots,P_N(t)]^T$ denote the instantaneous UAV transmit-power vector. The state component is defined as
\begin{equation}
s(t)=
\left[
\frac{\operatorname{vec}(\mathbf{X})}{D},
\frac{\operatorname{vec}(\mathbf{Q}(t))}{D/2},
\frac{\mathbf{P}(t)}{P_{\mathrm{unit}}}
\right]^T.
\end{equation}

Here, $\mathbf{X}$ contains the current horizontal GU coordinates, $D$ is the mission-area dimension, and $P_{\mathrm{unit}}$ denotes the power-unit conversion factor used in the implementation. The UAV position term uses absolute-coordinate scaling by the half-width of the service area rather than a GBS-relative coordinate transformation. The final term numerically converts the mW-valued transmit-power vector to watts through $P_{\mathrm{unit}}$. The critic therefore has access to the complete network geometry, the current joint control state, and the joint actions during centralized training, while actor execution remains decentralized.

\subsubsection{Continuous Action Space}

Each UAV uses a continuous four-dimensional action
$\mathbf a_n(t)=[a_n^x,a_n^y,a_n^z,a_n^P]^T$. The actor output is squashed by $\tanh$ to $[-1,1]$. The three motion components are first scaled by a configurable motion scale $S_{\rm move}$ and then projected onto the feasible three-dimensional displacement ball with radius $V_{\max}\delta_t$. Let
\begin{equation}
    \mathbf d_n^{\rm raw}(t)
    =
    S_{\rm move}
    [a_n^x(t),a_n^y(t),a_n^z(t)]^T.
\end{equation}
The implemented motion vector is
\begin{equation}
    \mathbf d_n(t)=
    \begin{cases}
        \mathbf d_n^{\rm raw}(t),
        & \|\mathbf d_n^{\rm raw}(t)\|_2
        \leq V_{\max}\delta_t,\\[1mm]
        V_{\max}\delta_t
        \dfrac{\mathbf d_n^{\rm raw}(t)}
        {\|\mathbf d_n^{\rm raw}(t)\|_2},
        & \text{otherwise},
    \end{cases}
\end{equation}
followed by clipping to the admissible mission area and altitude interval. The position update is
\begin{equation}
    \mathbf q_n(t+1)
    =
    \operatorname{clip}
    \left(\mathbf q_n(t)+\mathbf d_n(t)\right).
\end{equation}

The power action is scaled by the configured power adjustment step $\delta_p$ and clipped to the feasible transmit-power interval:
\begin{equation}
    P_n(t+1)
    =
    \operatorname{clip}
    \left(
    P_n(t)+\delta_p a_n^P(t),
    P_{\min},P_{\max}
    \right).
\end{equation}
\color{black}

\subsubsection{Multi-Objective Reward Mechanism}

\label{subsubsec:reward}

To align the agents' behavior with the system objectives, we design a comprehensive reward function. The instantaneous reward $r_n(t)$ for UAV agent $u_n$ at time step $t$ combines four weighted performance terms with implementation-specific coverage and constraint penalties:
\begin{equation}\label{eq:hybrid_reward}
\begin{split}
    r_n(t) = \ & w_1 r^{\text{cov}}(t) + w_2 r^{\text{ee}}_n(t) + w_3 r^{\text{load}}(t) \\
    & + w_4 r^{\text{rate}}(t) - w_5 r^{\text{pen}}_n(t) \\
    & + b_{\rm cov}\mathbb{I}(r^{\text{cov}}(t)\geq\tau_{\rm cov})
    -p_{\rm uns}(1-r^{\text{cov}}(t)).
\end{split}
\end{equation}
where $w_1,\ldots,w_5$ are positive weighting factors controlling the relative contributions of coverage, energy efficiency, load fairness, rate fairness, and constraint penalties, respectively. The first four terms instantiate the high-level objectives introduced in \eqref{eq:global_utility}, while the implemented reward additionally includes coverage shaping and constraint penalties.

Here, $b_{\mathrm{cov}}$, $p_{\mathrm{uns}}$, and $\tau_{\mathrm{cov}}$ denote the near-full-coverage bonus, unserved-user penalty, and coverage-shaping threshold, respectively; their numerical settings are reported in Table~\ref{tab:sim_params}. ORCHID adopts a progress-dependent weighting schedule during training rather than using fixed objective weights. Let $\tau=e/E_{\max}\in[0,1]$ denote the normalized training progress. The coverage weight $w_1$ remains fixed, the energy-efficiency weight $w_2$ increases with training progress, the fairness weights $w_3$ and $w_4$ decrease, and the common penalty scale $w_5$ increases. The corresponding endpoint settings are summarized in Table~\ref{tab:sim_params}.

\paragraph{Coverage Rate ($r^{\text{cov}}$)}
To minimize the number of users whose access links fall below the service threshold, this global reward component encourages the fleet to maximize the coverage ratio:
\begin{equation}\label{eq:cov_reward}
    r^{\text{cov}}(t)=\frac{|\mathcal{G}^{\rm cov}(t)|}{M},
\end{equation}
where $\mathcal{G}^{\rm cov}(t)$ is the set of users whose strongest
GBS/UAV link satisfies $\Gamma_{\rm req}$, as defined in
\eqref{eq:coverage_set}.
\color{black}

\paragraph{Energy Efficiency ($r^{\text{ee}}$)}
To promote energy-efficient operation, each UAV receives a normalized
per-UAV bit-per-Joule signal based on its associated traffic and total power:
\begin{equation}\label{eq:ee_reward}
    r^{\text{ee}}_n(t)=\operatorname{Norm}_n\!\left(
    \frac{\sum_{e_m\in\mathcal G_n(t)}R_m(t)}
    {P_{\text{total},n}(t)+\epsilon_{\rm ee}}\right).
\end{equation}
The implementation normalizes each EE stream using an exponential moving mean and variance, followed by standardization and clipping for stable learning; the corresponding normalization settings are reported in Table~\ref{tab:sim_params}. The system-level EE reported in the evaluation is computed separately from aggregate throughput and total UAV power.
The cumulative energy budget in~\eqref{const:energy} is a problem-level per-UAV battery-feasibility condition; the current implementation does not use it as a separate episode-rejection rule. The EE reward is therefore a performance-shaping signal that encourages lower energy consumption rather than a hard cumulative-energy enforcement mechanism.
\color{black}

\paragraph{UAV Load Fairness ($r^{\text{load}}$)}

To prevent load imbalance among UAVs, we apply JFI to the number of served users $K_n(t) = |\mathcal{G}_n(t)|$:
\begin{equation}\label{eq:load_fairness_reward}
    r^{\text{load}}(t) = \frac{\left(\sum_{n=1}^N K_n(t)\right)^2}{N \sum_{n=1}^N \left(K_n(t)^2+\epsilon_{\rm JFI}\right)},
\end{equation}
where $N$ is the total number of UAVs and $\epsilon_{\rm JFI}$ is a small numerical stabilizer. This term penalizes configurations where some UAVs are idle while others carry disproportionate numbers of associated users. As in the implementation, the load-JFI stream is normalized online and clipped using the reward-normalization settings in Table~\ref{tab:sim_params} before entering the reward.

\paragraph{User-Rate Fairness ($r^{\text{rate}}$)}
To evaluate user-level service equity consistently with the implementation, we apply JFI to all $M$ users, with unserved users contributing $R_m(t)=0$:
\begin{equation}\label{eq:rate_fairness_reward}
    r^{\text{rate}}(t)=\mathcal J(\mathbf R(t))=
    \frac{\left(\sum_{m=1}^{M}R_m(t)\right)^2}
    {M\sum_{m=1}^{M}\left(R_m^2(t)+\epsilon_{\rm JFI}\right)}.
\end{equation}
If all user rates are zero, the implementation sets the fairness term to zero. The JFI stream is normalized online using exponential moving statistics and clipped before entering the reward; the normalization settings are reported in Table~\ref{tab:sim_params}.
\color{black}

\paragraph{Constraint Penalties ($r^{\text{pen}}$)}
The implementation combines physical-boundary, power, collision, and backhaul-congestion penalties. In particular, the backhaul component depends on the excess load ratio rather than on a standalone SNR threshold:
\begin{equation}\label{eq:penalty_reward}
\begin{split}
    r^{\text{pen}}_n(t)=&\;r^{\text{pos}}_n(t)+r^{\text{pow}}_n(t)
    +\lambda_c\mathbb{I}\!\left(\min_{j\neq n}d_{nj}(t)<d_{\min}\right)\\
    &+\min\left\{\lambda_{\rm BH}[\rho_n^{\rm BH}(t)-1]_+,\rho_{\rm cap}\right\},
\end{split}
\end{equation}
where $r^{\text{pos}}_n(t)$ and $r^{\text{pow}}_n(t)$ denote the
implemented boundary and transmit-power violation terms, respectively, and $[x]_+=\max(x,0)$. The collision coefficient $\lambda_c$ is applied once whenever UAV $n$ has at least one separation violation, while $\lambda_{\rm BH}$ and $\rho_{\rm cap}$ denote the backhaul excess coefficient and penalty cap, respectively. The numerical values of these penalty parameters are reported in Table~\ref{tab:sim_params}; all are applied before the common penalty scale $w_5$.
\color{black}

\paragraph{Optimization Goal (Cumulative Return)}

The ultimate goal of the RL agent is to maximize the expected discounted cumulative return over the mission horizon $T$:
\begin{equation}\label{eq:cumulative_return}
    J(\pi_\theta) = \mathbb{E}_{\pi_\theta} \left[ \sum_{t=1}^{T} \gamma_{\mathrm{df}}^{t-1} r_n(t) \right],
\end{equation}
where $\pi_\theta$ is the policy parameterized by $\theta$, and $\gamma_{\mathrm{df}} \in [0, 1)$ is the discount factor determining the importance of future rewards.

\subsubsection{MAPPO Learning Objective}
To solve this optimization problem, the policy $\pi_{\theta}$ and centralized critic are updated iteratively. The actor is updated by maximizing the clipped surrogate objective:
\begin{equation}
\begin{split}
    \mathcal{L}(\theta) = \mathbb{E}_t \Big[ & \min \Big( r_t(\theta) \hat{A}_t, \\
    & \text{clip}(r_t(\theta), 1-\epsilon_{\text{clip}}, 1+\epsilon_{\text{clip}}) \hat{A}_t \Big) \Big] + c_e H(\pi_\theta),
\end{split}
\end{equation}
where $r_t(\theta)$ is the probability ratio between new and old policies, $\hat{A}_t$ is the \textit{Generalized Advantage Estimation} (GAE), $\epsilon_{\text{clip}}$ is the clipping parameter to limit policy updates, and $c_e$ is the entropy coefficient regulating the exploration-exploitation trade-off via the entropy term $H(\pi_\theta)$.

\subsection{Event-Triggered Stability Enhancement via Reset-and-Finetune}
\label{subsec:reset_finetune}

During Phase~II, ORCHID incorporates an event-triggered Reset-and-Finetune (R\&F) mechanism to improve late-stage training stability in on-policy multi-agent reinforcement learning. Rather than performing resets at fixed training milestones, R\&F evaluates recent training behavior every $W$ episodes and intervenes only when the learned policy exhibits sustained stagnation under both fairness and reward criteria.

Let $J_e$ and $U_e$ denote the episode-level average user-rate JFI and the total episode reward at episode $e$, respectively. The detector is evaluated every $W$ episodes. At each inspection episode, the most recent $L$ episodes form the trend window, which is partitioned into two consecutive sub-windows of length $W$ (with the reported setting satisfying $L=2W$). The mean JFI values of the first and second sub-windows are defined as
\begin{equation}
\bar{J}_{\mathrm{first}}(e)
= \frac{1}{W}\sum_{i=e-L+1}^{e-W} J_i, \quad
\bar{J}_{\mathrm{last}}(e)
= \frac{1}{W}\sum_{i=e-W+1}^{e} J_i.
\label{eq:jfi_windows}
\end{equation}

The relative change in fairness is then measured by
\begin{equation}
\label{eq:jfi_change}
\Delta_J(e)
=
\frac{
\left|
\bar{J}_{\mathrm{last}}(e)
-
\bar{J}_{\mathrm{first}}(e)
\right|
}{
\bar{J}_{\mathrm{first}}(e)+\epsilon_{\mathrm{JFI}}
},
\end{equation}
where $\epsilon_{\mathrm{JFI}}$ is a small positive constant introduced
for numerical stability. A small $\Delta_J(e)$ indicates that the
fairness level has remained approximately unchanged over the recent
training period.

In parallel, the reward trajectory within the latest $L$ episodes is examined using an ordinary least-squares linear fit
\begin{equation}
U_i=\beta_1 i+\beta_0+\varepsilon_i,
\end{equation}
where $\beta_1$ and $\beta_0$ denote the fitted slope and intercept, respectively. To obtain a scale-normalized measure of the recent reward trend, we define
\begin{equation}\label{eq:reward_trend}
T_U(e)
=
\frac{|\hat{\beta}_1|}
{\operatorname{SE}(\hat{\beta}_1)},
\end{equation}
where $\hat{\beta}_1$ is the estimated slope and $\operatorname{SE}(\hat{\beta}_1)$ is its standard error. A smaller $T_U(e)$ indicates a flatter recent reward trajectory. The R\&F detector uses the predefined practical threshold $\tau_T$ to identify such a flat trend.

A reward-floor condition is also imposed to avoid triggering R\&F when the recent training trajectory is already at an undesirable reward level. Specifically, the mean reward of the latest sub-window is defined as
\begin{equation}\label{eq:reward_floor}
\bar{U}_{\mathrm{last}}(e)
=
\frac{1}{W}
\sum_{i=e-W+1}^{e} U_i .
\end{equation}

An R\&F event is triggered only when all required conditions are simultaneously satisfied:
\begin{equation}
\begin{aligned}
\mathrm{RF}(e)=1
\Longleftrightarrow\;&
\left(e>E_{\mathrm{warm}}\right)
\land
\left(e-e_{\mathrm{last}}\ge C_{\mathrm{rf}}\right) \\
&\land\left(\Delta_J(e)<\epsilon_J\right)
\land\left(T_U(e)<\tau_T\right) \\
&\land\left(\bar{U}_{\mathrm{last}}(e)>U_{\min}\right).
\end{aligned}
\label{eq:rf_trigger}
\end{equation}
where $E_{\mathrm{warm}}$ is the initial warm-up period, $C_{\mathrm{rf}}$ is the minimum separation between two consecutive R\&F events, $\epsilon_J$ is the fairness-stability tolerance, $\tau_T$ is the reward-trend threshold, $U_{\min}$ is the reward floor, and $e_{\mathrm{last}}$ records the episode of the most recent R\&F event.

When the trigger condition is satisfied, the current actor and centralized critic are rolled back to the best-total-reward checkpoint observed during training:
\begin{equation}\label{eq:rf_rollback}
(\theta,\phi)
\leftarrow
(\theta_{\mathrm{best}},\phi_{\mathrm{best}}).
\end{equation}
The Adam optimizers, which maintain first- and second-moment estimates of the gradients~\cite{reddi2018convergence}, are then re-initialized to remove their accumulated states. Let $\eta_{\ell}^{(r^-)}$ denote the learning rate of component $\ell\in\{\mathrm{actor},\mathrm{critic}\}$ immediately before the $r$-th R\&F event. The post-reset learning rate is updated as
\begin{equation}\label{eq:rf_lr}
\eta_{\ell}^{(r)}
=
\kappa\,\eta_{\ell}^{(r^-)} ,
\qquad
\ell\in\{\mathrm{actor},\mathrm{critic}\},
\end{equation}
where $\kappa$ is the R\&F learning-rate scaling factor.

The actor exploration scale (i.e., the log-standard-deviation parameter) is also reset to the predefined value reported in Table~\ref{tab:sim_params}. Before the first R\&F intervention, the optimizer learning rates follow the initial scheduling policy. Once an R\&F event is performed, the initial schedulers are disabled, and subsequent learning-rate adjustments are governed by the R\&F mechanism.

Through the combination of fairness stabilization, reward-trend monitoring, checkpoint rollback, optimizer reinitialization, and learning-rate scaling, R\&F provides a conservative mechanism for recovering from late-stage training stagnation while avoiding unnecessary periodic resets. 
\color{black}

\subsection{Algorithm Description and Complexity Analysis}
\label{subsec:algorithm_complexity}

The complete training and execution procedure of the proposed ORCHID framework is summarized in Algorithm~\ref{alg:orchid}. The learning process combines episode-wise topology initialization, MAPPO updates, and event-triggered R\&F interventions; the trained policy is then used for online decentralized inference.

\begin{algorithm2e}[!t]
%
\caption{ORCHID: Two-Stage Orchestration Framework}
\label{alg:orchid}
\KwIn{
Number of UAVs $N$, maximum episodes $E_{\max}$, horizon $T$,
and R\&F parameters
$(L,W,E_{\mathrm{warm}},C_{\mathrm{rf}},
\epsilon_J,\tau_T,U_{\min},\kappa)$
}
\KwOut{Trained actor policy $\pi_\theta$}

Initialize actor $\pi_{\theta}$, action-conditioned centralized critic $C_{\phi}$, rollout buffer $\mathcal{B}$, best-total-reward checkpoint $(\theta_{\mathrm{best}},\phi_{\mathrm{best}})$, best reward $U_{\mathrm{best}}$, and $e_{\mathrm{last}}\leftarrow 0$;

\For{$e\leftarrow1$ \KwTo $E_{\max}$}{

    Regenerate the user topology\;

    \tcp{Phase I: K-Means++ clustering with GBS-aware filtering}
    Apply K-Means++ clustering with GBS-aware filtering and assign the resulting
    UAV centroids\;

    Reset the environment and obtain $s_0$\;

    \tcp{Phase II: MAPPO training}
    \For{$t\leftarrow1$ \KwTo $T$}{
        Each UAV $n$ observes $o_n(t)$ and samples
        $a_n(t)\sim\pi_\theta(\cdot\mid o_n(t))$\;
        Execute the joint action and store the transition in $\mathcal{B}$\;
    }
    
    Update $(\theta,\phi)$ using the MAPPO objective\;
    Clear the rollout buffer $\mathcal{B}$\;
    Compute episode reward $U_e$ and average JFI $J_e$\;

    \If{$U_e>U_{\mathrm{best}}$}{
        Save $(\theta,\phi)$ as $(\theta_{\mathrm{best}},\phi_{\mathrm{best}})$\;
        $U_{\mathrm{best}}\leftarrow U_e$\;
    }

    \tcp{Event-triggered R\&F detection}
    \If{an R\&F inspection is due and the latest $L$ episodes are available}{
        Compute $\Delta_J(e)$, $T_U(e)$, and
        $\bar{U}_{\mathrm{last}}(e)$ according to
        \eqref{eq:jfi_change}, \eqref{eq:reward_trend}, and
        \eqref{eq:reward_floor}\;

        \If{the trigger condition in \eqref{eq:rf_trigger} is satisfied}{
            Roll back $(\theta,\phi)$ according to \eqref{eq:rf_rollback}\;
            Re-initialize the Adam optimizers\;
            Update the learning rates according to \eqref{eq:rf_lr}\;
            Disable the initial LR schedulers and reset the actor
            exploration scale\;
            $e_{\mathrm{last}}\leftarrow e$\;
        }
    }
}
\Return{$\pi_\theta$}\;
\end{algorithm2e}

\subsubsection{Offline Complexity (Initialization \& Training)}
The computational burden is primarily concentrated in the pre-deployment phase. For each training episode, Phase I applies the K-Means++ clustering followed by GBS-aware filtering to the regenerated user topology. Its computational complexity is approximately $\mathcal{O}(I(N+1)M)$, where $I$ is the number of clustering iterations, $N+1$ is the number of candidate centroids, and $M$ is the number of users. Since the topology is regenerated and the initialization is repeated in every episode, the overall Phase-I complexity scales as $\mathcal{O}(E_{\max}I(N+1)M)$.

Phase II is dominated by environment rollouts and MAPPO optimization. Let $\mathcal{C}_{\rm roll}$ denote the computational cost of one rollout step per UAV and $\mathcal{C}_{\rm upd}$ denote the corresponding per-sample update cost of the actor and critic. Let $E_{\rm PPO}$ denote the number of PPO optimization epochs per update. The resulting complexity can be expressed as $\mathcal{O}\!\left(E_{\max}TN\mathcal{C}_{\rm roll}+E_{\max}E_{\rm PPO}TN\mathcal{C}_{\rm upd}\right)$. The PPO batch size determines the mini-batch organization and memory requirement; for a fixed number of collected samples, it does not introduce an additional linear factor in the leading-order computation.

The R\&F detector is evaluated every $W$ episodes during training. For each inspection, the latest $L$ episode-level reward and JFI observations are examined to compute the fairness-change metric and the OLS-based reward-trend statistic. A direct implementation therefore requires $\mathcal{O}(L)$ operations per inspection, resulting in an aggregate detection cost of approximately $\mathcal{O}((E_{\max}/W)L)$ when inspections are performed every $W$ episodes. When the trigger condition is satisfied, checkpoint rollback and optimizer re-initialization operate on the actor and critic parameters. If $P_\theta$ and $P_\phi$ denote the numbers of actor and critic parameters, respectively, the intervention cost is $\mathcal{O}(P_\theta+P_\phi)$ per R\&F event. Since the number of interventions is bounded by the cooldown constraint, the cumulative
R\&F intervention cost is at most $\mathcal{O}((E_{\max}/C_{\rm rf})(P_\theta+P_\phi))$.

Thus, R\&F adds only training-time overhead determined by the detector window and actor/critic parameter counts; no R\&F computation is required during deployment. All training and R\&F operations are performed offline on a centralized server prior to deployment.

\subsubsection{Online Policy-Inference Complexity}
During deployment, each UAV independently executes one forward pass through its actor network. For an MLP with network depth $D_{\rm net}$ and layer widths $H_l$, the per-decision complexity is $\mathcal{O}(\sum_{l=1}^{D_{\rm net}-1} H_lH_{l+1})$. This computation is independent across UAVs and involves only the local observation and actor network. Hence, the R\&F mechanism does not introduce additional inference-time computation or latency.

\section{Analytical Foundation of Fairness and Training Stability}
\label{sec:theoretical_analysis}

In this section, we provide an analytical and mechanism-level
interpretation of the main components underlying the proposed framework. First, we formulate intra-cell bandwidth allocation as an \textit{Egalitarian Bargaining Game} (EBG) and derive the closed-form max-min solution used by the MMF configuration. We then interpret the induced Max-SNR user association as a state-dependent many-to-one matching process to clarify its interaction with UAV-load fairness. Finally, we provide a qualified optimization-level interpretation of the learning-rate scaling used by R\&F. System-level performance trade-offs are evaluated empirically in Section~\ref{sec:evaluation}.

\subsection{Egalitarian Bargaining for Bandwidth Allocation}

Consider the resource contention within a specific UAV cell $u_n$ at time $t$. Given the trajectory $\mathbf{q}_n(t)$ and transmit power $P_n(t)$ determined by the MAPPO agent, the SNR $\gamma_{n,m}(t)$ of each associated user $e_m \in \mathcal{G}_n(t)$ is fixed for the current state. The remaining allocation variable is the continuous bandwidth fraction $\omega_{n,m}(t)$. For notational simplicity, the time index is omitted in the following derivation.

For the bargaining subproblem, we take the achievable rate of each associated user as its utility. Specifically, for user $e_m\in\mathcal{G}_n(t)$,
\begin{equation}
U_m(\omega_{n,m})
=
\omega_{n,m} B
\log_2(1+\gamma_{n,m}),
\end{equation}
where $B$ denotes the total access bandwidth available to UAV $u_n$, and $\omega_{n,m}$ is the fraction assigned to user $e_m$.

With zero disagreement utility, the standard Nash Bargaining Solution (NBS) maximizes the product of user utilities, or equivalently the sum of their logarithms~\cite{nash1950bargaining,kalai1977proportional}. Under the present fixed-SNR, linear-in-bandwidth rate model, the factors $B\log_2(1+\gamma_{n,m})$ are fixed with respect to $\boldsymbol{\omega}$. The NBS therefore reduces to maximizing the product of the bandwidth fractions subject to the bandwidth constraint, which yields equal bandwidth sharing. Accordingly, we use PF to denote this equal-sharing configuration in the experiments, while the MAPPO policy continues to optimize UAV trajectories and transmit powers through the composite reward.

We formulate the intra-cell bandwidth allocation as an EBG \cite{kalai1977proportional}, where the associated users are treated as bargaining participants with zero disagreement utility. The resulting egalitarian solution maximizes the minimum user utility and therefore captures the intended max-min rate-equality principle:
\begin{subequations}
\begin{align}
    \max_{\boldsymbol{\omega}} \quad
    & \min_{e_m \in \mathcal{G}_n(t)}
    U_m(\omega_{n,m}) \\
    \text{s.t.}\quad
    & \sum_{e_m \in \mathcal{G}_n(t)}
    \omega_{n,m} \leq 1, \\
    & \omega_{n,m} \geq 0,
    \quad \forall e_m.
\end{align}
\end{subequations}

\begin{proposition}\label{prop:ebg_solution}
Assuming positive spectral efficiency for all associated users, the unique egalitarian max-min solution assigns a bandwidth fraction $\omega_{n,m}^{*}$ inversely proportional to the spectral efficiency of user $e_m$. Specifically,
\begin{equation}\label{eq:ebg_optimal_omega}
    \omega_{n,m}^{*}
    =
    \frac{
    \left[\log_2(1+\gamma_{n,m})\right]^{-1}
    }{
    \sum_{e_j \in \mathcal{G}_n(t)}
    \left[\log_2(1+\gamma_{n,j})\right]^{-1}
    }.
\end{equation}
\end{proposition}

\begin{proof}

Since $B>0$ is common to all users, it does not affect the maximizing bandwidth allocation. Let $\zeta_m=\log_2(1+\gamma_{n,m})>0$ and define
\[
r_{\min}
=
\min_{e_m\in\mathcal{G}_n(t)}
\omega_{n,m}\zeta_m.
\]
For any feasible allocation, $\omega_{n,m}\zeta_m\geq r_{\min}$ for every associated user, which implies
\[
1
\geq
\sum_{e_m\in\mathcal{G}_n(t)}\omega_{n,m}
\geq
r_{\min}
\sum_{e_m\in\mathcal{G}_n(t)}\zeta_m^{-1}.
\]
Hence,
\[
r_{\min}
\leq
\left(
\sum_{e_m\in\mathcal{G}_n(t)}\zeta_m^{-1}
\right)^{-1}.
\]
This upper bound is attained by
\[
\omega_{n,m}^{*}
=
\frac{\zeta_m^{-1}}
{\sum_{e_j\in\mathcal{G}_n(t)}\zeta_j^{-1}},
\]
for which all associated users attain the same rate. Equality requires $\omega_{n,m}\zeta_m=r_{\min}$ for every user together with full use of the available bandwidth, yielding the unique maximizing allocation in \eqref{eq:ebg_optimal_omega}.

\end{proof}

\begin{remark}[Role of the EBG Solution]
The closed-form solution in Proposition~\ref{prop:ebg_solution} is used as the deterministic intra-cell bandwidth allocation rule for the MMF configuration in the implemented environment. The same derivation applies to the GBS service cell by replacing $B$ with $B_{\rm GBS}$ and using the corresponding terrestrial-link SNRs. In contrast, the PF configuration uses equal bandwidth sharing, corresponding to the NBS solution of the fixed-SNR linear-in-bandwidth subproblem. In both configurations, bandwidth allocation is determined by the environment rather than included in the MAPPO action space. The learned policy therefore optimizes UAV trajectories and transmit powers, while bandwidth allocation follows the selected fairness configuration.
\end{remark}

\begin{remark}[Channel-Heterogeneity Interpretation]
Proposition~\ref{prop:ebg_solution} shows that users with lower spectral efficiency receive larger bandwidth fractions under the MMF configuration. This channel-dependent allocation provides the mechanism-level basis for the served-user bottleneck-rate protection evaluated in Section~\ref{sec:strategy_fairness}, without implying that MMF necessarily improves every aggregate fairness metric.
\end{remark}

\subsection{Matching-Based Interpretation of User Association}

The association between the service nodes $\mathcal{U}_0$ and GUs $\mathcal{G}$ can be interpreted as a state-dependent many-to-one matching process induced by the Max-SNR association rule in \eqref{eq:assoc_rule}. This interpretation is introduced only to clarify the interaction among UAV control, user association, and fleet-load fairness; it is not a formal matching-equilibrium or stable-matching analysis.

\begin{remark}[Matching-Based Interpretation]
At each time slot, each GU is associated with its strongest eligible service node according to the current channel state. Because UAV positions and transmit powers affect these channel states, the learned policy indirectly reshapes the resulting association pattern. At the same time, the UAV-load JFI term in the reward discourages persistent concentration of associated users on a subset of UAVs. The resulting interpretation is consistent with the load-balancing objective and with the UAV-load fairness trends observed in the experiments.
\end{remark}

\subsection{Optimization Interpretation of Training Stability}
\label{subsec:stability_analysis}

Following the preceding fairness-oriented analysis, the remaining challenge is stable policy optimization in a highly dynamic and non-convex multi-agent environment. Standard MARL algorithms may suffer from premature convergence or late-stage policy degradation under such conditions, motivating the proposed R\&F mechanism as a stability-oriented optimization intervention.

Let $\boldsymbol{\theta}$ denote the parameters of the actor network. During Phase II training, the policy is updated using the Adam optimizer, which maintains the exponential moving averages of the gradient $\mathbf{m}_t$ (first moment) and the squared gradient $\mathbf{v}_t$ (second moment). The parameter update rule is given by:
\begin{equation}
    \boldsymbol{\theta}_{t+1} = \boldsymbol{\theta}_t -
\frac{\eta}{\sqrt{\hat{\mathbf{v}}_t} + \epsilon_{\rm opt}}
\hat{\mathbf{m}}_t,
\end{equation}
where $\eta$ is the learning rate, $\hat{\mathbf{m}}_t$ and
$\hat{\mathbf{v}}_t$ are the bias-corrected moments, and
$\epsilon_{\rm opt}>0$ is the optimizer numerical constant. 

As the UAV swarm gradually forms a stable topology, the composite reward may approach a local plateau. In this regime, optimizer-state statistics accumulated during earlier exploration can become mismatched with the later fine-tuning dynamics. The R\&F mechanism intervenes when the multi-metric trigger condition defined in Section~\ref{subsec:reset_finetune} is satisfied.

At each R\&F event, the current actor and critic are first rolled back to the best-total-reward checkpoint from the preceding training history. The Adam optimizers are then re-initialized, clearing their accumulated first- and second-moment states, while the learning rates are reduced multiplicatively according to~\eqref{eq:rf_lr}. The actor exploration scale (i.e., the log-standard-deviation parameter) is also reset to the predefined value reported in Table~\ref{tab:sim_params}. This procedure is therefore a controlled restart around a previously strong policy, followed by lower-step-size fine-tuning; it is not a formal guarantee of escaping every local optimum.

Under a simplified update-scaling model that isolates the learning-rate effect from optimizer-state reset and changing gradients, the parameter-increment variance scales approximately as 
\begin{equation}\label{eq:rf_variance}
    \operatorname{Var}(\Delta\boldsymbol{\theta}') \approx
    \kappa^2\operatorname{Var}(\Delta\boldsymbol{\theta}).
\end{equation}
For the configured learning-rate scaling factor reported in Table~\ref{tab:sim_params}, this simplified relation illustrates the quadratic dependence of the parameter-increment variance on the learning-rate scaling. It provides only an optimization-level interpretation and does not constitute a formal convergence or variance bound for the full MAPPO dynamics, which also depend on optimizer reset, checkpoint rollback, newly sampled gradients, PPO updates, and the non-stationary multi-agent environment.

\section{Performance Evaluation}
\label{sec:evaluation}

This section evaluates the proposed ORCHID framework through simulations. We detail the simulation environment, comparative baselines, and the metrics used to quantify coverage, energy efficiency, and fairness.

\subsection{Simulation Setup}
\label{subsec:setup}

The default simulation configuration is summarized in Table~\ref{tab:sim_params}. During training, user topologies are regenerated episode-wise according to the clustered and spatially dispersed sampling scheme defined in Section~\ref{subsec:user_dist}. The standard protocol uses $E_{\max}$ training episodes with per-episode horizon $T$ and five independent training seeds. During inference-based evaluations, exploratory noise and parameter updates are disabled.

In our experiments, we contrast two distinct bandwidth configurations to investigate their interaction with service equity and system efficiency. The higher-level MAPPO reward structure is otherwise kept unchanged; we refer to the two bandwidth configurations as \textit{MMF-oriented} (MMF) and \textit{Proportional Fairness} (PF), respectively:
\begin{itemize}
    \item \textbf{MMF-oriented (MMF):} The default configuration for ORCHID. It uses the closed-form EBG allocation in each service cell, while the MAPPO policy optimizes trajectories and transmit powers using the composite fairness-aware reward.
    \item \textbf{PF:} The comparative configuration using equal bandwidth sharing in each service cell. Under the fixed-channel linear-in-bandwidth model, this is also the proportional-fair/NBS bandwidth solution; the same MAPPO reward structure is retained.
\end{itemize}
Unless otherwise specified, the simulations employ the MMF configuration.

\subsection{Comparative Baselines}
\label{subsec:baselines}

Furthermore, we compare ORCHID with one ablation variant and three representative baselines to evaluate the contributions of initialization, bandwidth configuration, and learning architecture. Conventional iterative optimization methods, such as successive convex approximation, are not included because their runtime is costly for repeated online decisions; instead, static K-Means is adopted as the representative heuristic. For learning-based comparisons, we consider the off-policy deterministic CTDE architecture MADDPG and the on-policy stochastic CTDE architecture MAPPO, with MAPPO represented by the w/o-R\&F ablation. The configurations are summarized below:

\begin{itemize}
    \item \textbf{Ablation (w/o R\&F):} Uses the same GBS-aware K-Means initialization, MMF configuration, and MAPPO architecture as ORCHID, with the R\&F intervention disabled.

    \item \textbf{Baseline 1 (MADDPG + PF):} Uses MADDPG~\cite{lowe2017multi} with the PF configuration and the same GBS-aware K-Means initialization.

    \item \textbf{Baseline 2 (MADDPG + MMF):} Uses MADDPG with the MMF configuration and the same GBS-aware K-Means initialization.

    \item \textbf{Baseline 3 (Static K-Means + PF):} Uses the Phase~I GBS-aware K-Means deployment and keeps both UAV locations and transmit powers fixed throughout the mission, without learning-based trajectory or power adaptation.
\end{itemize}

\color{black}
For consistency, power and SNR quantities used in the analytical expressions are represented in linear units (W, W/Hz, and linear SNR), while engineering values are reported in dBm/dB or mW in Table~\ref{tab:sim_params}. The access threshold $\Gamma_{\rm req}$ is converted to its linear representation when substituted into the corresponding equations. Distances and coordinates in the analytical model are expressed in meters, with their numerical settings reported in Table~\ref{tab:sim_params}.

\begin{table}[!ht]
\centering
\caption{Simulation Parameters}
\label{tab:sim_params}
\setlength{\tabcolsep}{4pt}
\resizebox{\columnwidth}{!}{%
\begin{tabular}{|l|l|c|}
\hline
\textbf{Category} & \textbf{Parameter} & \textbf{Value} \\ \hline\hline

\multirow{8}{*}{Network}
& Service Area Size ($D \times D$) & $2,000 \times 2,000$ m$^2$ \\
& GBS Position ($\mathbf{q}_0$)
& $[1000,1000]$ m \\
& Number of GUs ($M$) & 100 \\
& Cluster Hotspots ($K_{\rm cl}$) & 5 \\
& Users per Hotspot ($M_{\rm cl}$) & 20 \\
& Clustered/Uniform Prob. & $0.5 / 0.5$ \\
& User Scattering Std. Dev. ($\sigma_{\rm scatter}$) & 150 m \\
& Time Slot Duration ($\delta_t$) & 1 s \\ \hline

\multirow{10}{*}{UAV}
& Number of UAVs ($N$) & 6 \\
& Motion Scaling Factor ($S_{\rm move}$)
& 20 m \\
& Initial Cruising Altitude ($H_{\rm init}$) & 100 m \\
& Altitude Range $[H_{\min},H_{\max}]$ & $[80,120]$ m \\
& Maximum UAV Speed ($V_{\max}$) & 15 m/s \\
& Min. Safety Clearance ($d_{\min}$) & 5 m \\
& Transmit Power Range $[P_{\min},P_{\max}]$ & $[100,199.53]$ mW \\
& Power Conversion Factor ($P_{\rm unit}$)
& $10^3$ \\
& Power Adjustment Step Size ($\delta_p$) & 9.95 mW \\
& Per-UAV Energy Budget ($E_{\rm budget}$)
& 100 kJ \\ \hline

\multirow{8}{*}{Aerodynamics}
& Blade Profile Power ($P_0$) & 79.856 W \\
& Induced Power ($P_i$) & 314.15 W \\
& Tip Speed of Rotor Blade ($U_{\rm tip}$) & 120.0 m/s \\
& Mean Rotor Induced Velocity ($v_0$) & 4.03 m/s \\
& Fuselage Drag Ratio ($d_{\rm fuse}$) & 0.6 \\
& Air Density ($\rho$) & 1.225 kg/m$^3$ \\
& Rotor Solidity ($s$) & 0.05 \\
& Rotor Disc Area ($A$) & 0.503 m$^2$ \\ \hline

\multirow{11}{*}{Channel}
& Access Carrier Frequency ($f_c$)
& 2.4 GHz \\
& Backhaul Carrier Frequency ($f_{\rm BH}$)
& 28 GHz \\
& Access Bandwidth ($B$)
& 10 MHz \\
& GBS Bandwidth ($B_{\rm GBS}$)
& 10 MHz \\
& Backhaul Bandwidth ($B_{\rm BH}$) & 20 MHz \\
& GBS Transmit Power ($P_{\rm GBS}$) & 30 dBm \\
& Noise Power Density ($N_0$) & $-174$ dBm/Hz \\
& Access-Link SNR Threshold ($\Gamma_{\rm req}$) & 6 dB \\
& Optimal Elevation Angle ($\vartheta_{\rm opt}$) & $42.44^\circ$ \\
& S-curve Parameters ($a,b$) & 9.61, 0.16 \\
& Excessive Path Losses ($\eta_{\rm LoS},\eta_{\rm NLoS}$) & 1.0 dB, 20.0 dB \\ \hline

\multirow{16}{*}{Learning}
& Discount Factor ($\gamma_{\rm df}$) & 0.99 \\
& Actor Learning Rate ($\eta_{\rm actor}$) & $10^{-4}$ \\
& Critic Learning Rate ($\eta_{\rm critic}$) & $10^{-3}$ \\
& Batch Size ($N_{\rm batch}$) & 128 \\
& PPO Optimization Epochs ($E_{\rm PPO}$) & 4 \\
& PPO Clip Ratio ($\epsilon_{\rm clip}$) & 0.2 \\
& GAE Parameter ($\lambda_{\rm GAE}$) & 0.95 \\
& Entropy Coefficient ($c_e$) & Dynamic ($0.01 \to 0$) \\
& Actor Hidden Layers / Activation
& $256\times256\times256$ / ReLU \\
& Critic Hidden Layers / Activation
& $256\times256\times128$ / ReLU \\
& Reward Norm. Momentum ($\beta_{\rm norm}$) & 0.95 \\
& Reward Norm. Clip Range & $[-3,3]$ \\
& Initial StepLR Step Size
& 100 episodes \\
& Initial StepLR Decay Factor
& 0.9 \\
& Total Steps per Training Episode ($T$)
& 200 \\
& Total Training Episodes ($E_{\max}$)
& 700 \\ \hline

\multirow{10}{*}{R\&F}
& Trend Window Length ($L$)
& 100 episodes \\
& Sub-window / Check Interval ($W$)
& 50 episodes \\
& Warm-up Episodes ($E_{\rm warm}$)
& 200 \\
& R\&F Cooldown ($C_{\rm rf}$)
& 100 episodes \\
& JFI Stability Tolerance ($\epsilon_J$)
& 0.05 \\
& Reward Trend Threshold ($\tau_T$)
& 1.5 \\
& Reward Floor ($U_{\min}$)
& $-10.0$ \\
& Learning-Rate Scaling Factor ($\kappa$)
& 0.5 \\
& Actor log-Std. Reset
& $-0.2$ \\
& JFI Numerical Stabilizer ($\epsilon_{\rm JFI}$)
& $10^{-6}$ \\ \hline

\multirow{13}{*}{Reward}
& Coverage-Term Weight ($w_1$)
& 0.5 \\
& EE Objective Weight ($w_2$)
& Dynamic ($0.2 \to 0.5$) \\
& Load JFI Objective Weight ($w_3$)
& Dynamic ($0.8 \to 0.5$) \\
& Rate JFI Objective Weight ($w_4$)
& Dynamic ($0.8 \to 0.5$) \\
& Penalty Scale Weight ($w_5$)
& Dynamic ($0.1 \to 0.3$) \\
& Collision Penalty Coefficient ($\lambda_c$)
& 2.5 \\
& Boundary Penalty Coefficient ($\lambda_b$)
& 1.0 \\
& Power Violation Coefficient
& 1.0 \\
& Near-full Coverage Bonus ($b_{\rm cov}$)
& 0.5 \\
& Unserved-User Penalty ($p_{\rm uns}$)
& 1.0 \\
& Coverage Shaping Threshold ($\tau_{\rm cov}$)
& 0.99 \\
& Backhaul Excess Coefficient ($\lambda_{\rm BH}$)
&  $0.25$ \\ 
& Backhaul Excess Penalty Cap
($\rho_{\rm cap}$)
&  $45$ \\ \hline

\end{tabular}%
}
\end{table}

\subsection{Evaluation Metrics}
\label{subsec:metrics}

The system performance is evaluated using the following three categories of metrics. Unless otherwise specified, scalar metrics are averaged over the evaluation episodes:

\subsubsection{User Coverage Rate}
This metric quantifies the service availability of the network. Under the orthogonal access assumptions above, a user is considered covered if the access-link SNR exceeds the minimum service threshold $\Gamma_{\text{req}}$. The coverage rate is defined as:

\begin{equation}
    \text{Coverage}=\frac{1}{T}\sum_{t=1}^{T}
    \frac{|\mathcal G^{\rm cov}(t)|}{M}.
\end{equation}
Here, $\mathcal G^{\rm cov}(t)$ is the set of users whose strongest GBS/UAV
link satisfies $\Gamma_{\rm req}$, as defined in~\eqref{eq:coverage_set}.
\color{black}

\subsubsection{System Energy Efficiency}
Given the limited energy budget of UAVs, energy efficiency serves as a key performance metric. We define the reported system energy efficiency as the ratio of aggregate UAV-served throughput to total UAV power consumption, including both transmission and propulsion. Grid-powered GBS traffic is excluded from this metric because it is not part of the UAV energy budget. This metric measures instantaneous energy efficiency and should be distinguished from the cumulative battery-feasibility condition in \eqref{const:energy}. The instantaneous system energy efficiency $\text{EE}_{\text{sys}}(t)$ (bits/Joule) is:
\begin{equation}
    \text{EE}_{\text{sys}}(t) = \frac{\sum_{u_n \in \mathcal{U}} \sum_{e_m \in \mathcal{G}_n(t)} R_m(t)}{\sum_{u_n \in \mathcal{U}} P_{\text{total},n}(t)}.
\end{equation}
The episode-level system EE reported in the experiments is the time average $T^{-1}\sum_{t=1}^{T}\text{EE}_{\text{sys}}(t)$.


\subsubsection{Fairness and Load Balancing Metrics}
To evaluate the equity of resource allocation over the entire mission duration, we adopt the Time-Averaged JFI, extending the instantaneous definitions provided in Section~\ref{subsec:phase2}.

\paragraph{UAV Load Fairness ($\overline{\text{JFI}}_{\text{load}}$)}
To assess the workload distribution among the fleet, we compute the time-averaged UAV Load Fairness based on the instantaneous definition in~\eqref{eq:load_fairness_reward}, denoted as $r^{\text{load}}(t)$:
\begin{equation}
    \overline{\text{JFI}}_{\text{load}} = \frac{1}{T} \sum_{t=1}^T r^{\text{load}}(t).
\end{equation}
Ideally, $\overline{\text{JFI}}_{\text{load}} \approx 1$ indicates that the workload remains evenly distributed among the UAV fleet during the entire mission, preventing scenarios where some UAVs are overloaded while others are idle.

\paragraph{User Rate Fairness ($\overline{\text{JFI}}_{\text{rate}}$)}
Complementing load balance, we evaluate the service consistency experienced by the GUs. Based on the instantaneous User Rate Fairness $r^{\text{rate}}(t)$ defined in~\eqref{eq:rate_fairness_reward}, the time-averaged metric is:
\begin{equation}
    \overline{\text{JFI}}_{\text{rate}} = \frac{1}{T} \sum_{t=1}^T r^{\text{rate}}(t).
\end{equation}
A higher $\overline{\text{JFI}}_{\text{rate}}$ indicates a more even rate distribution across all users, with unserved users represented by zero rate. The coverage ratio separately reports service availability, so the two metrics distinguish availability from rate equity.

\subsection{Convergence and Stability Analysis}
\label{sec:convergence}

To evaluate algorithmic convergence and training stability, Fig.~\ref{fig:convergence} presents the learning trajectories of four key metrics over the standard training horizon. In Fig.~\ref{fig:convergence:total_reward}, the two MADDPG baselines exhibit clear late-stage degradation, whereas the MAPPO-based ORCHID and its w/o-R\&F ablation remain within a bounded near-stationary regime. Although Baseline~3 attains higher shaped reward and fairness indices, its static deployment avoids action-induced boundary, power, and collision penalties and is therefore better viewed as a geometric reference than as a directly comparable learned mobile policy.

Among the learned mobile policies, ORCHID achieves the highest system EE in Fig.~\ref{fig:convergence:nee}, reaching approximately $5.2\times10^5$~bits/J, while the MADDPG baselines remain substantially lower. Fig.~\ref{fig:convergence:rate_jfi} shows that ORCHID maintains competitive user-rate fairness, with Baseline~2 attaining slightly higher values during part of the late-stage regime. By contrast, Fig.~\ref{fig:convergence:load_jfi} reveals a clearer fleet-level advantage: ORCHID and the w/o-R\&F variant maintain substantially higher UAV-load fairness than the MADDPG baselines. These results indicate that ORCHID's main advantage lies in energy efficiency and load balancing rather than universal dominance in every fairness metric.

Comparison with the w/o-R\&F ablation isolates the role of Reset-and-Finetune. The two MAPPO variants follow similar trajectories during early and middle training because the detector remains inactive until the prescribed multi-metric plateau conditions are satisfied. Once triggered, R\&F rolls back to the best-total-reward checkpoint, re-initializes the Adam optimizers, and resumes fine-tuning with the configured reduced learning rate. Because the trigger is event-driven rather than scheduled at a fixed episode, its activation time may vary across training runs. The resulting trajectory shows improved late-stage stability compared with the w/o-R\&F ablation, supporting R\&F as a conservative stabilization mechanism rather than a periodic reset strategy.

\begin{figure}[!t]
    \vspace{-.5em}
    \centering
    \subfloat[Convergence of Total Reward.
    \label{fig:convergence:total_reward}]{
        \includegraphics[width=0.48\linewidth]{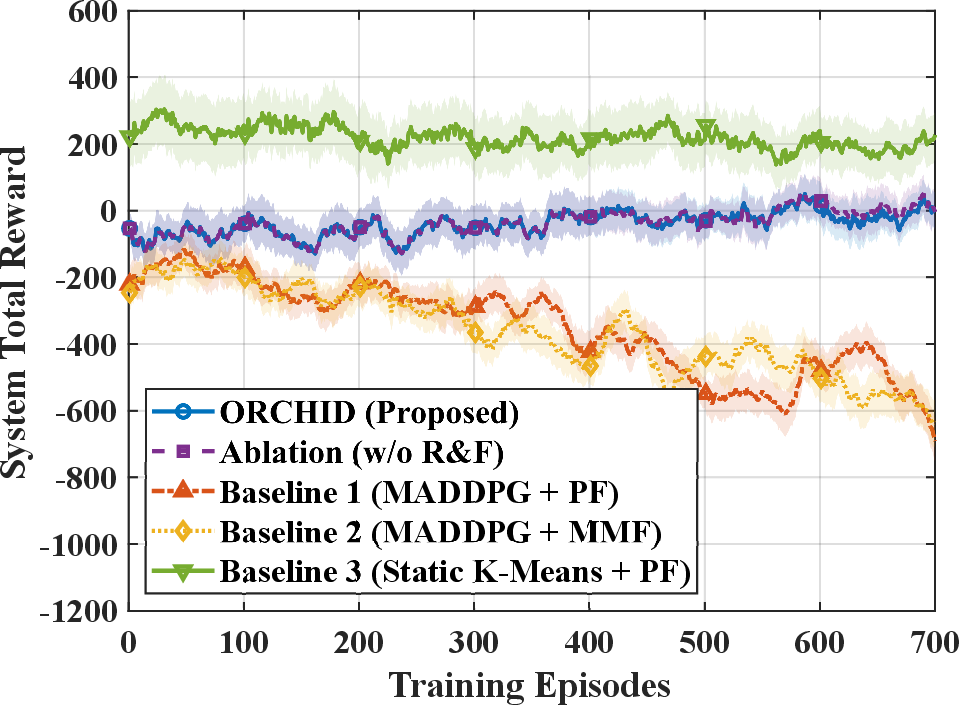}
    }
    \hfill
    \subfloat[Convergence of System EE.
    \label{fig:convergence:nee}]{
        \includegraphics[width=0.46\linewidth]{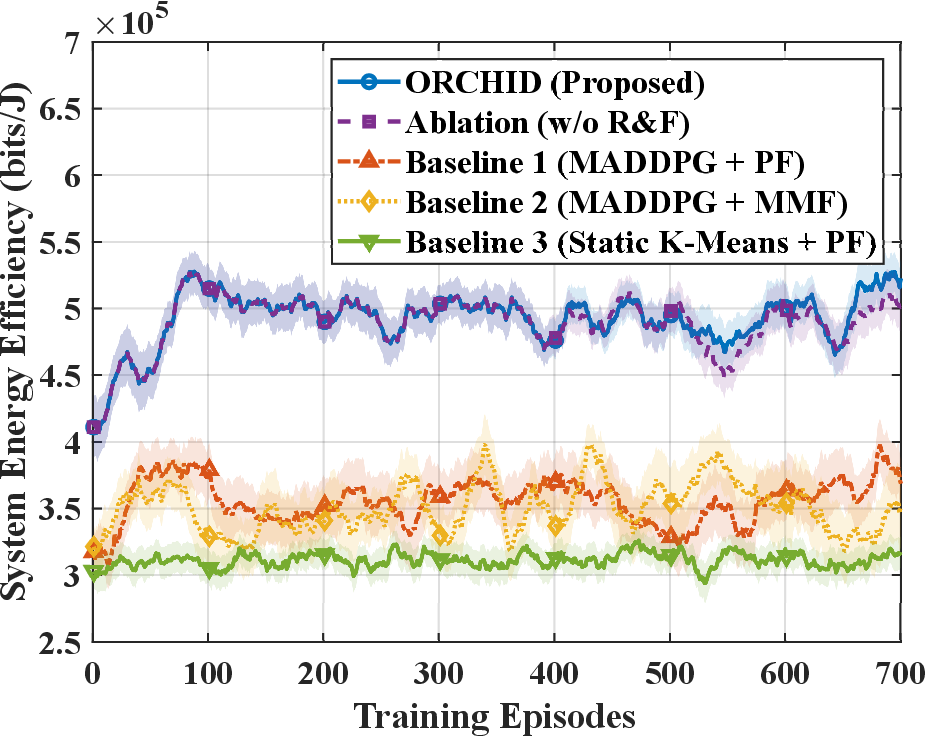}
    }

    \subfloat[Convergence of User Rate JFI.
    \label{fig:convergence:rate_jfi}]{
        \includegraphics[width=0.47\linewidth]{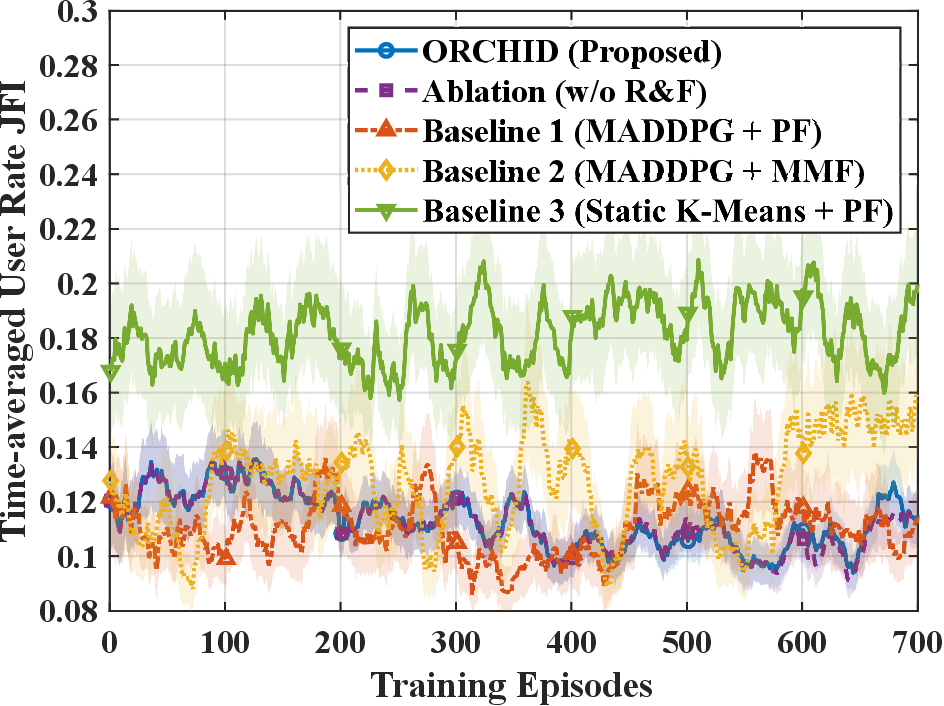}
    }
    \hfill
    \subfloat[Convergence of UAV Load JFI.
    \label{fig:convergence:load_jfi}]{
        \includegraphics[width=0.47\linewidth]{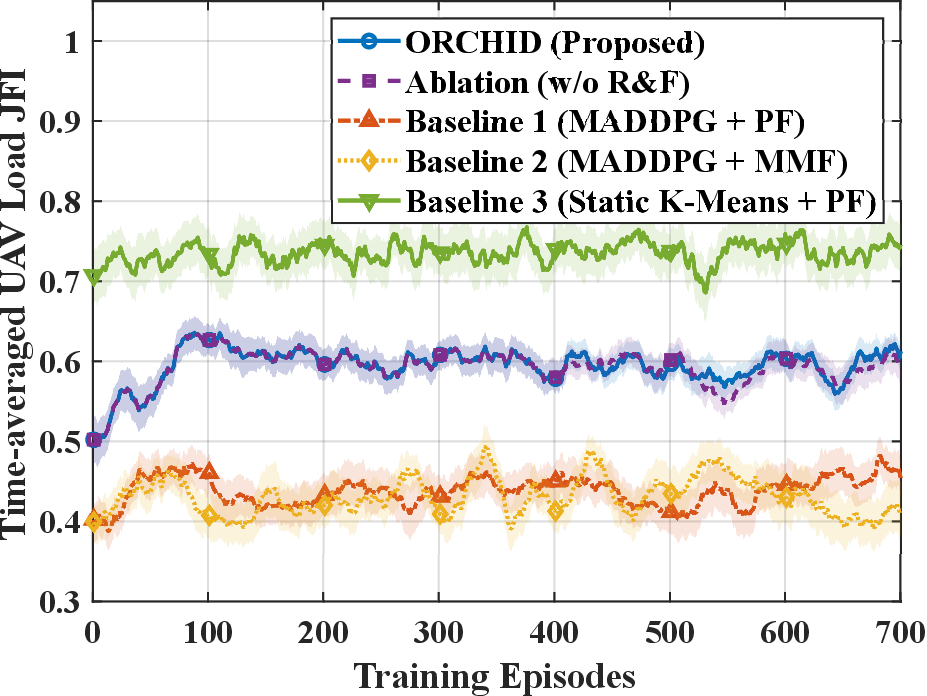}
    }
    \caption{Convergence and stability analysis over the standard training schedule. The curves show the mean over five independent training seeds, while the shaded regions denote one standard deviation.}
    \label{fig:convergence}
\end{figure}

\subsection{Scalability Analysis}
\label{sec:scalability}

To evaluate zero-shot scalability, the policy trained under the default $N=6$, $M=100$ configuration in Table~\ref{tab:sim_params} is directly tested under different user populations and swarm sizes without retraining. All metrics are averaged over 100 independent inference episodes per configuration.

\subsubsection{Impact of Number of Users}

Fig.~\ref{fig:scalability_users} evaluates system performance as the ground-user population increases from 50 to 200. As shown in Fig.~\ref{fig:scalability_users:cov}, all methods maintain relatively stable coverage across the tested loads. The learned mobile policies maintain coverage of approximately 78\% to 80\%, while Baseline~3 remains slightly higher at approximately 81\% to 84\%. This indicates that service availability remains broadly preserved as user density changes beyond the training configuration.

Fig.~\ref{fig:scalability_users:nee} shows that system EE increases with user population for all approaches because aggregate served throughput grows while per-UAV power remains bounded by the adopted motion and transmit-power limits. ORCHID exhibits the strongest scaling trend, increasing from approximately $2.9\times10^5$~bits/J at 50 users to about $5.1\times10^5$~bits/J at 200 users. Its advantage over the MADDPG baselines becomes more pronounced at larger user populations, indicating a more favorable trade-off between throughput and energy efficiency under heavier traffic loads.

The fairness metrics reveal a different scaling behavior. In Fig.~\ref{fig:scalability_users:rate_jfi}, user-rate JFI decreases as more users share fixed bandwidth resources. ORCHID remains competitive with the MADDPG baselines, although Baseline~1 attains slightly higher rate JFI at larger user populations. In contrast, Fig.~\ref{fig:scalability_users:load_jfi} shows that UAV-load fairness improves as user density increases, and ORCHID consistently achieves the highest load JFI among the learned mobile policies. Overall, ORCHID preserves competitive coverage and user-level fairness while providing clear advantages in system EE and fleet-load balancing under user-population scaling.

\begin{figure}[!t]
    \vspace{-.5em}
    \centering
    \subfloat[Impact on system coverage ratio.
    \label{fig:scalability_users:cov}]{
        \includegraphics[width=0.47\linewidth]{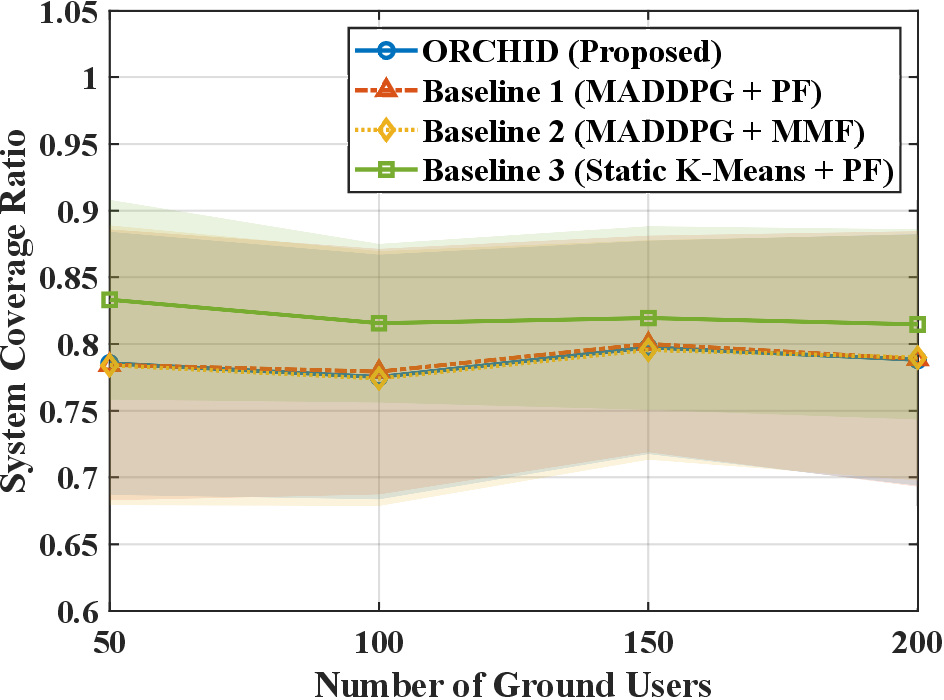}
    }
    \hfill
    \subfloat[Impact on system EE.
    \label{fig:scalability_users:nee}]{
        \includegraphics[width=0.47\linewidth]{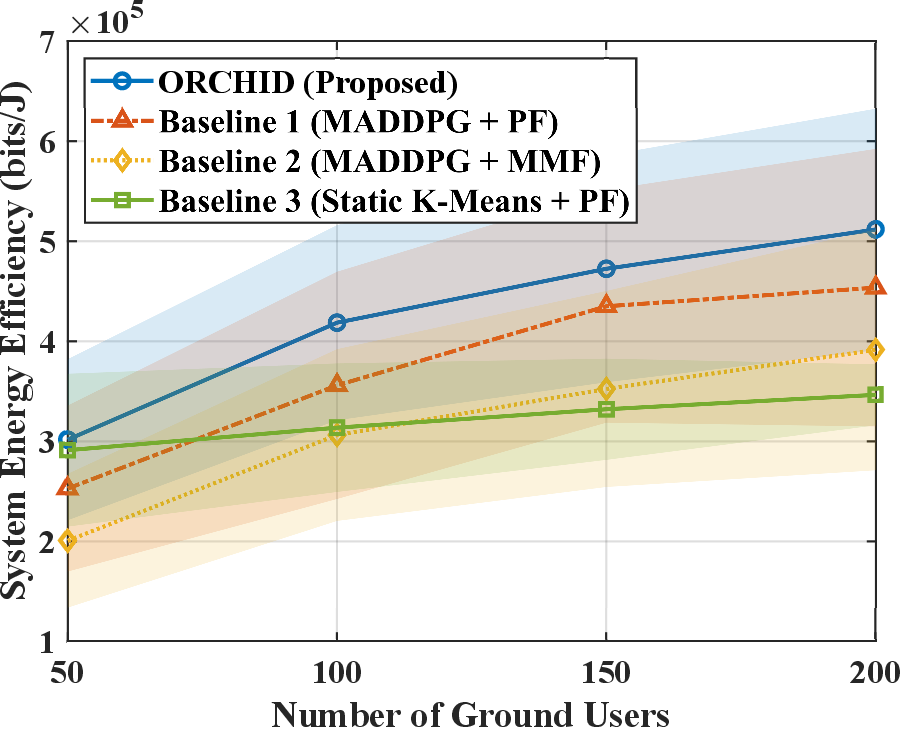}
    }

    \subfloat[Impact on User Rate JFI.
    \label{fig:scalability_users:rate_jfi}]{
        \includegraphics[width=0.47\linewidth]{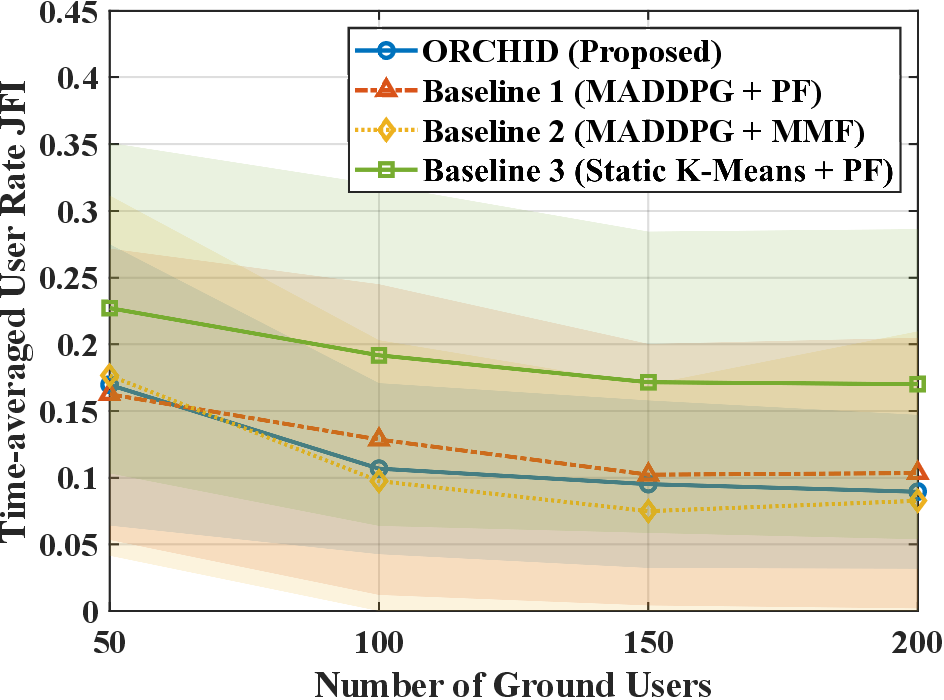}
    }
    \hfill
    \subfloat[Impact on UAV Load JFI.
    \label{fig:scalability_users:load_jfi}]{
        \includegraphics[width=0.47\linewidth]{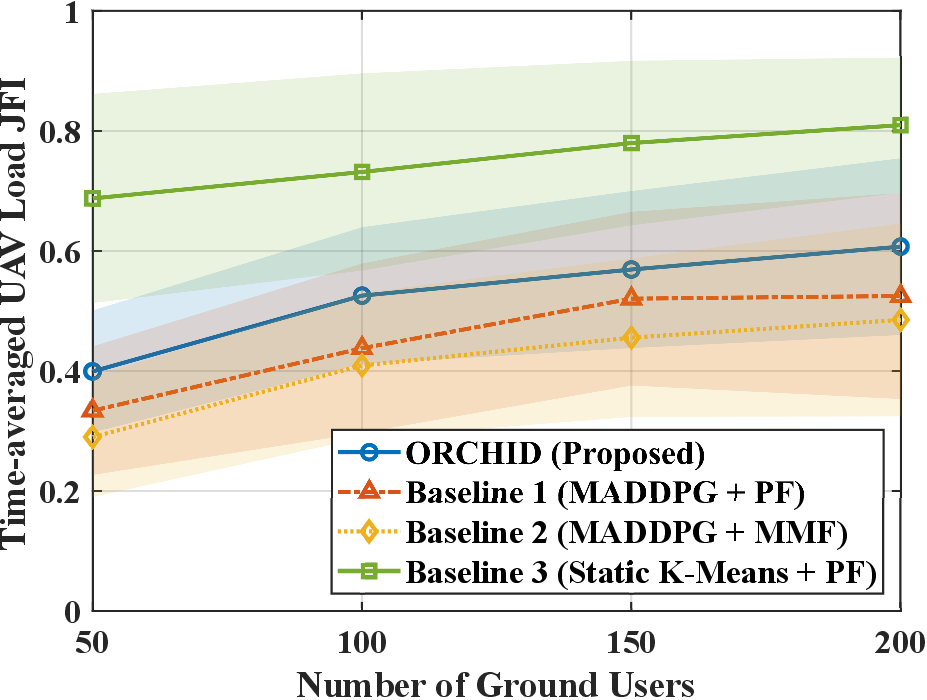}
    }
    \caption{Scalability analysis against the number of GUs. The curves show the mean over 100 independent inference episodes, and the shaded regions denote one standard deviation.}
    \label{fig:scalability_users}
\end{figure}

\subsubsection{Impact of Number of UAVs}

Fig.~\ref{fig:scalability_uavs} evaluates zero-shot scalability as the aerial swarm size increases from 4 to 10 UAVs. As shown in Fig.~\ref{fig:scalability_uavs:cov}, the learned mobile policies maintain stable coverage at approximately 78\% to 80\% across all tested swarm sizes. Baseline~3 increases from about 80\% to 85\% because the denser static deployment provides a finer spatial partition of the service area. Overall, the coverage results indicate that the learned policies preserve comparable service availability even when the number of UAVs differs from the training configuration.

Fig.~\ref{fig:scalability_uavs:nee} shows a clearer difference in system energy efficiency. ORCHID decreases only mildly from approximately $4.2\times10^5$~bits/J at $N=4$ to $3.95\times10^5$~bits/J at $N=10$, while remaining consistently above both MADDPG baselines. The corresponding EE values of Baselines~1 and~2 decline more noticeably as the swarm expands. This trend reflects the increasing aggregate UAV power consumption as additional aerial nodes are introduced, while ORCHID retains a more favorable trade-off between throughput and energy efficiency under zero-shot swarm scaling.

The fairness results exhibit a different pattern. In Fig.~\ref{fig:scalability_uavs:rate_jfi}, ORCHID remains comparable to Baseline~1 in user-rate JFI and generally above Baseline~2. Thus, the proposed method preserves competitive user-level fairness but does not universally dominate this metric. By contrast, Fig.~\ref{fig:scalability_uavs:load_jfi} shows a clear fleet-level advantage: although load JFI decreases mildly as more UAVs share a fixed user population, ORCHID consistently achieves the highest UAV-load fairness among the learned mobile policies. These results indicate that ORCHID is primarily distinguished by higher energy efficiency and more balanced multi-UAV load allocation, while maintaining comparable coverage and user-rate fairness.

\begin{figure}[!t]
    \vspace{-.5em}
    \centering
    \subfloat[Impact on system coverage ratio.
    \label{fig:scalability_uavs:cov}]{
        \includegraphics[width=0.47\linewidth]{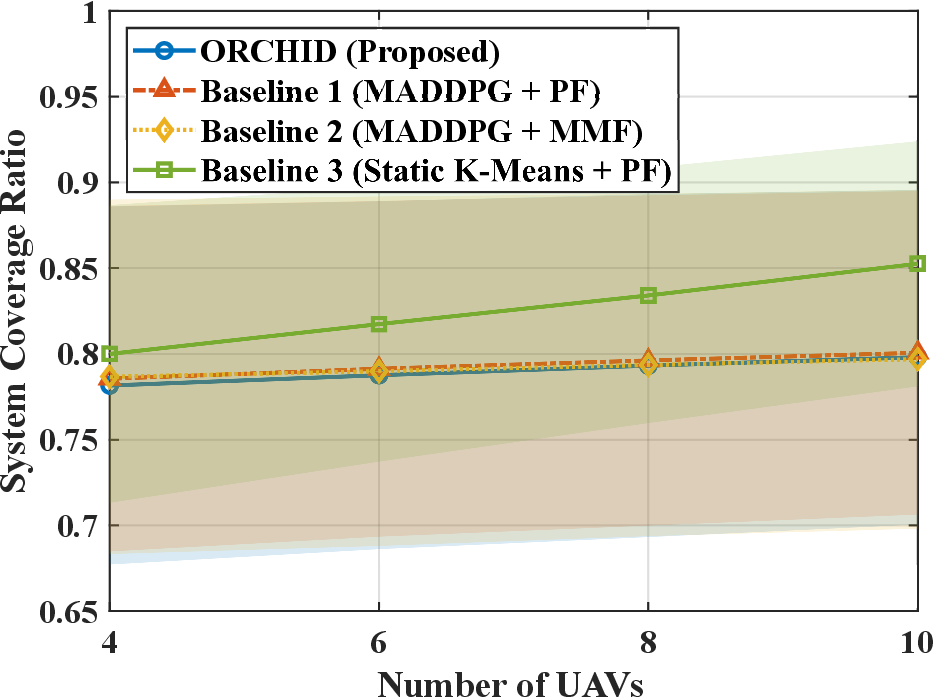}
    }
    \hfill
    \subfloat[Impact on system EE.
    \label{fig:scalability_uavs:nee}]{
        \includegraphics[width=0.47\linewidth]{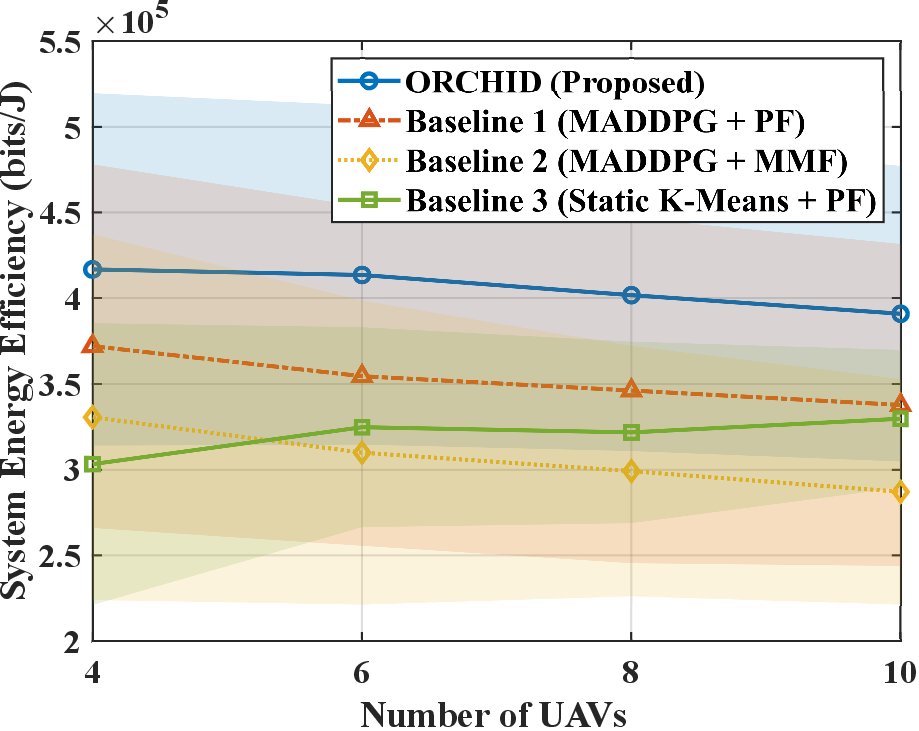}
    }

    \subfloat[Impact on User Rate JFI.
    \label{fig:scalability_uavs:rate_jfi}]{
        \includegraphics[width=0.47\linewidth]{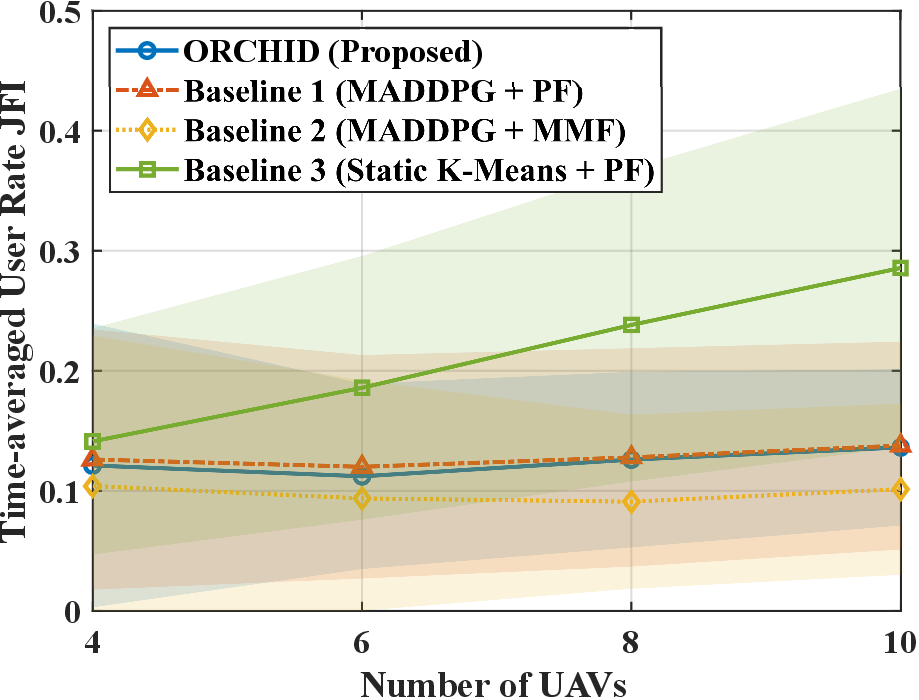}
    }
    \hfill
    \subfloat[Impact on UAV Load JFI.
    \label{fig:scalability_uavs:load_jfi}]{
        \includegraphics[width=0.47\linewidth]{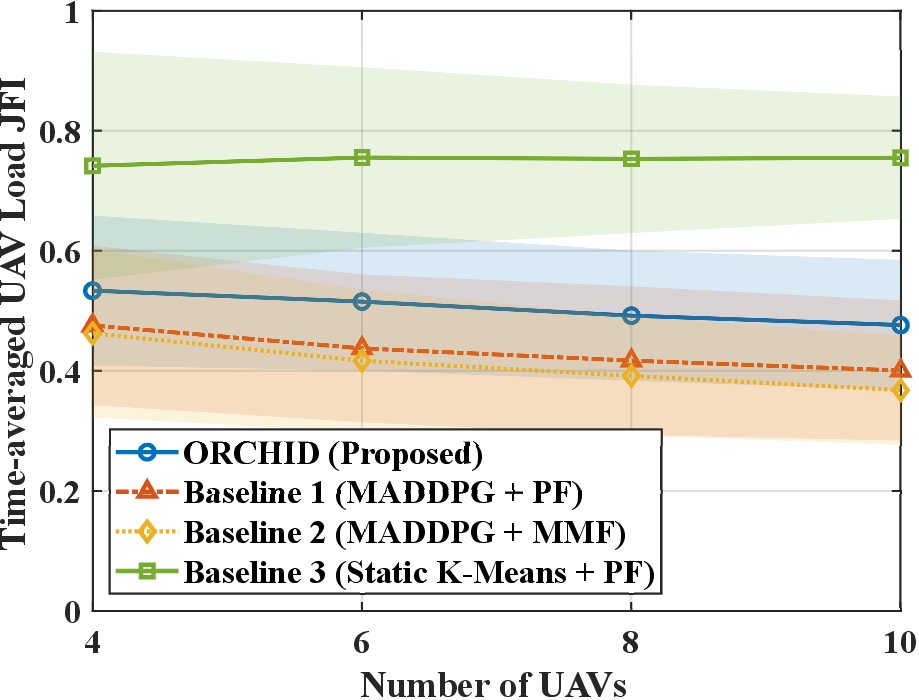}
    }
    \caption{Scalability analysis against the number of UAVs. The curves show the mean over 100 independent inference episodes, and the shaded regions denote one standard deviation.}
    \label{fig:scalability_uavs}
\end{figure}

\subsection{Generalization to Dynamic Environments via Zero-shot Policy Deployment}
\label{sec:mobility_inference}

To evaluate mobility generalization, the pre-trained policy is directly deployed, without fine-tuning or retraining, in environments governed by a discrete-time Random Walk (RW) mobility model with boundary reflection. Ground-user speed is varied from 0 to 15~m/s, with stochastic direction updates. All metrics are averaged over 100 independent inference episodes per speed level.

As shown in Fig.~\ref{fig:mobility:cov}, coverage remains comparatively stable across the tested mobility levels. The learned mobile policies remain close to 79\%, while the initial coverage advantage of Baseline~3 gradually diminishes as user speed increases. This indicates that the learned policies preserve service availability under user motion not encountered during training.

Fig.~\ref{fig:mobility:nee} reveals a stronger distinction in system energy efficiency. ORCHID consistently achieves the highest EE over the full mobility range, remaining above both MADDPG baselines and peaking near $4.5\times10^5$~bits/J around 10~m/s. In contrast, the static baseline degrades steadily as its fixed deployment becomes increasingly mismatched to the evolving user topology. The result suggests that ORCHID retains a more favorable energy-efficiency profile under dynamic spatial conditions.

The fairness trends further illustrate this robustness. User-rate JFI decreases for all methods as mobility increases, but ORCHID remains competitive and becomes the strongest learned mobile policy at the two highest tested speeds, 10 and 15~m/s, slightly surpassing both MADDPG baselines and approaching the static baseline at the highest tested speed. For UAV-load fairness, ORCHID consistently outperforms both MADDPG baselines and improves moderately with user speed, whereas the static baseline declines. Hence, under unseen mobility conditions, ORCHID maintains competitive coverage and user-level fairness while preserving stronger EE and fleet-load balancing.

\begin{figure}[!t]
    \vspace{-.5em}
    \centering
    \subfloat[Impact on system coverage ratio.
    \label{fig:mobility:cov}]{
        \includegraphics[width=0.47\linewidth]{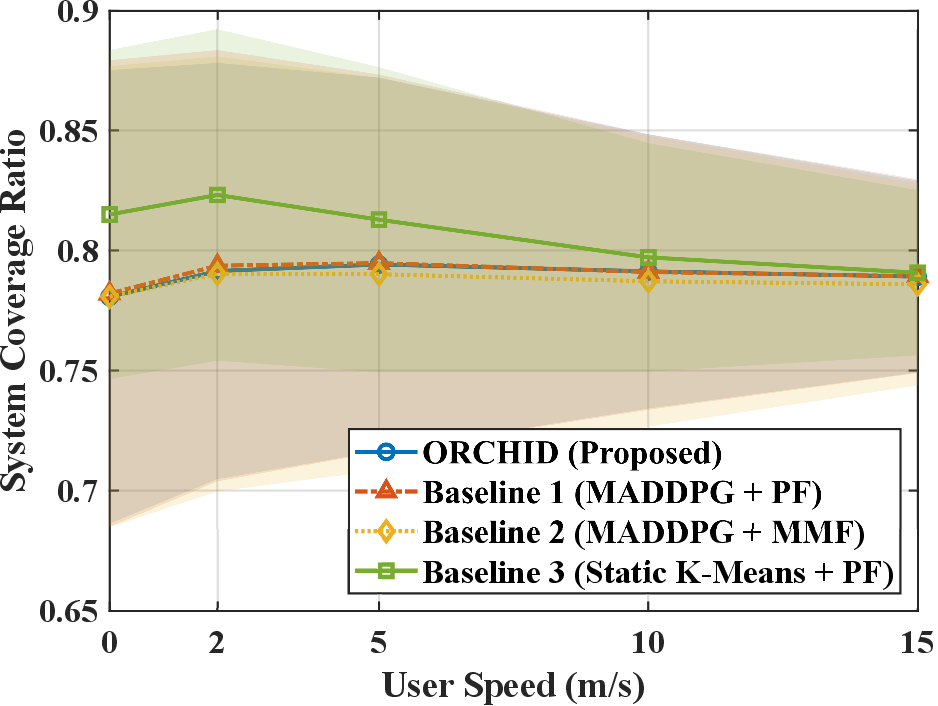}
    }
    \hfill
    \subfloat[Impact on system EE.
    \label{fig:mobility:nee}]{
        \includegraphics[width=0.47\linewidth]{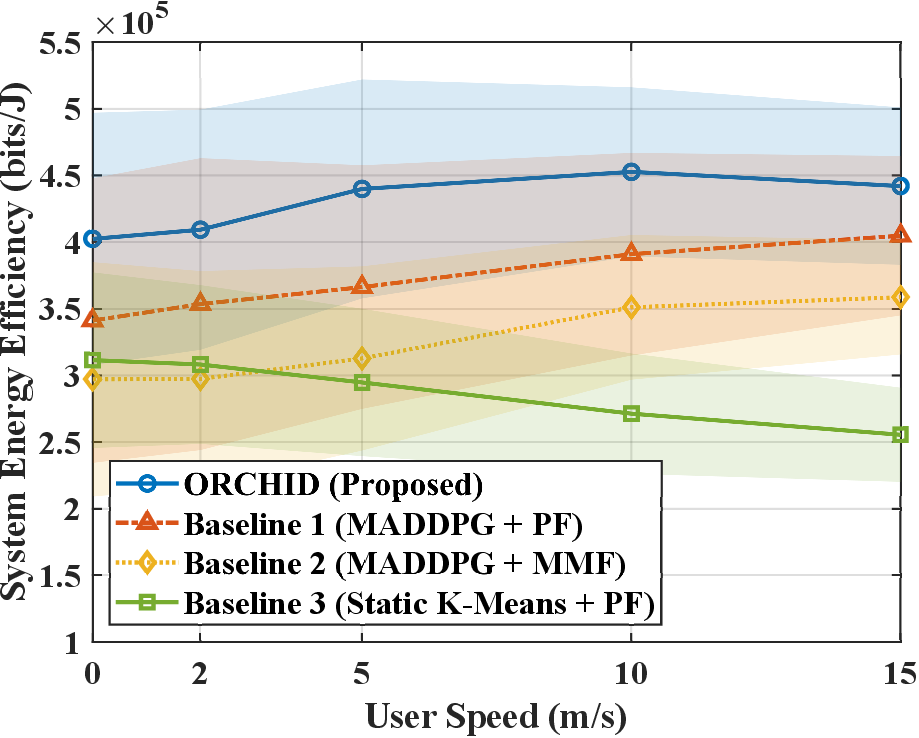}
    }

    \subfloat[Impact on User Rate JFI.
    \label{fig:mobility:rate_jfi}]{
        \includegraphics[width=0.47\linewidth]{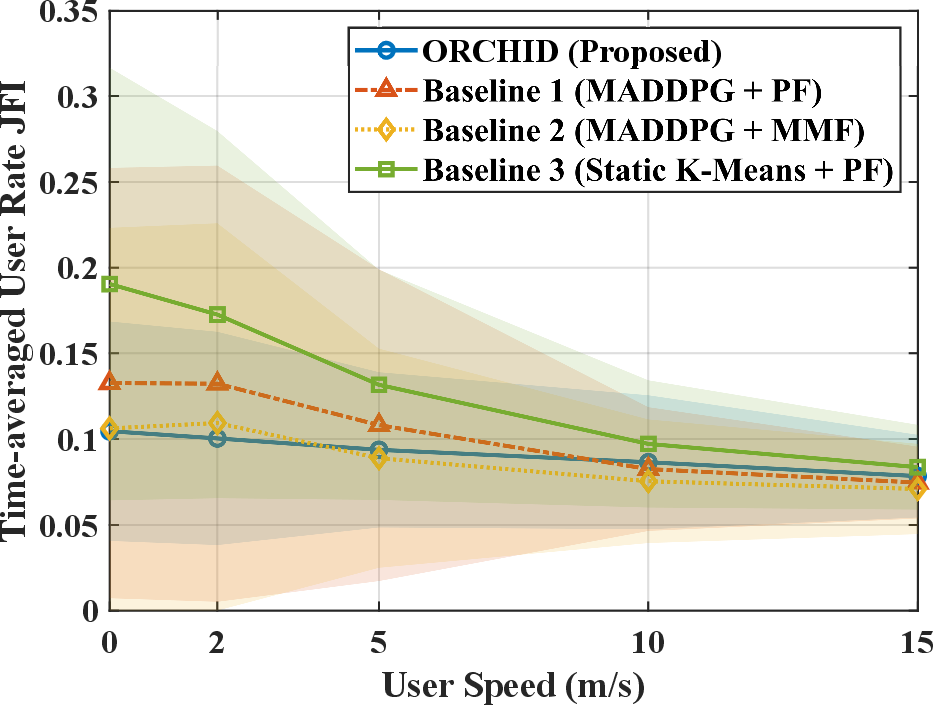}
    }
    \hfill
    \subfloat[Impact on UAV Load JFI.
    \label{fig:mobility:load_jfi}]{
        \includegraphics[width=0.47\linewidth]{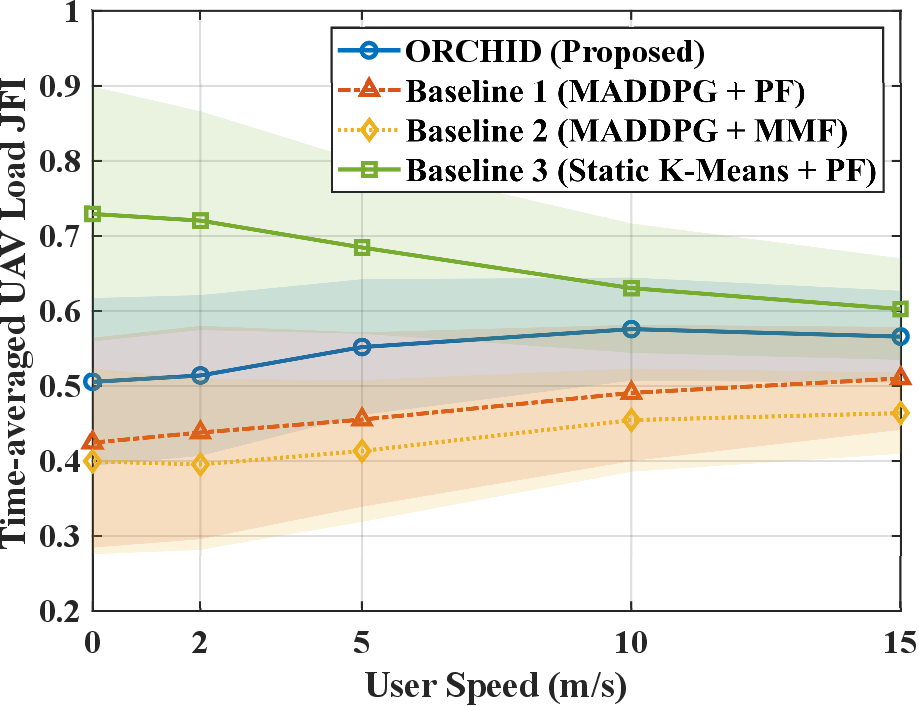}
    }
    \caption{Robustness analysis against ground-user mobility. The curves show the mean over 100 independent inference episodes for each tested speed level, while the shaded regions denote one standard deviation.}
    \label{fig:mobility}
\end{figure}

\subsection{Ablation Study on Warm-Start Initialization}
\label{sec:ablation_init}

Fig.~\ref{fig:ablation_init} evaluates the contribution of the GBS-aware K-Means warm start by comparing it with random initialization under otherwise identical MAPPO training. The random variant can attain a higher shaped reward during parts of training, but with substantially larger variability, illustrating that shaped reward alone does not fully reflect physical service quality.

The physical metrics provide a clearer comparison. ORCHID reaches approximately $5.0\times10^5$~bits/J within 100 episodes, compared with about $3.3\times10^5$~bits/J under random initialization. The two variants exhibit comparable user-rate JFI, whereas the K-Means warm start achieves markedly higher UAV-load fairness, reaching approximately 0.61 versus 0.41 for random initialization. These results support the role of topology-aware initialization as an effective geometric warm start that guides subsequent MARL optimization toward more energy-efficient and load-balanced configurations.

\begin{figure}[!t]
    \vspace{-.5em}
    \centering
    \subfloat[Convergence of Total Reward.
    \label{fig:ablation_init:reward}]{
        \includegraphics[width=0.48\linewidth]{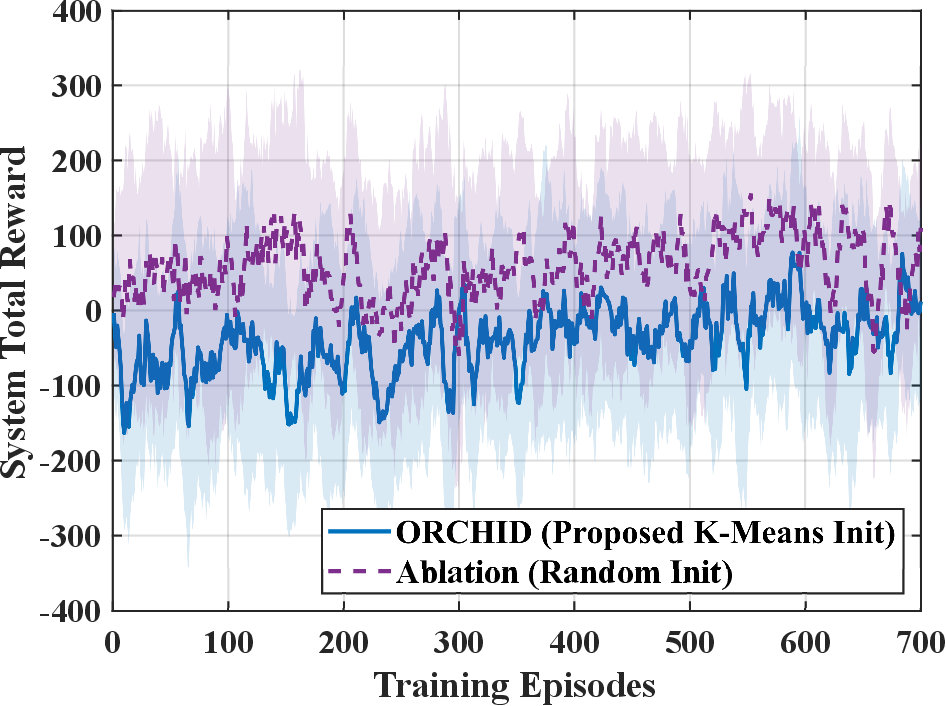}
    }
    \hfill
    \subfloat[Convergence of System EE.
    \label{fig:ablation_init:ee}]{
        \includegraphics[width=0.46\linewidth]{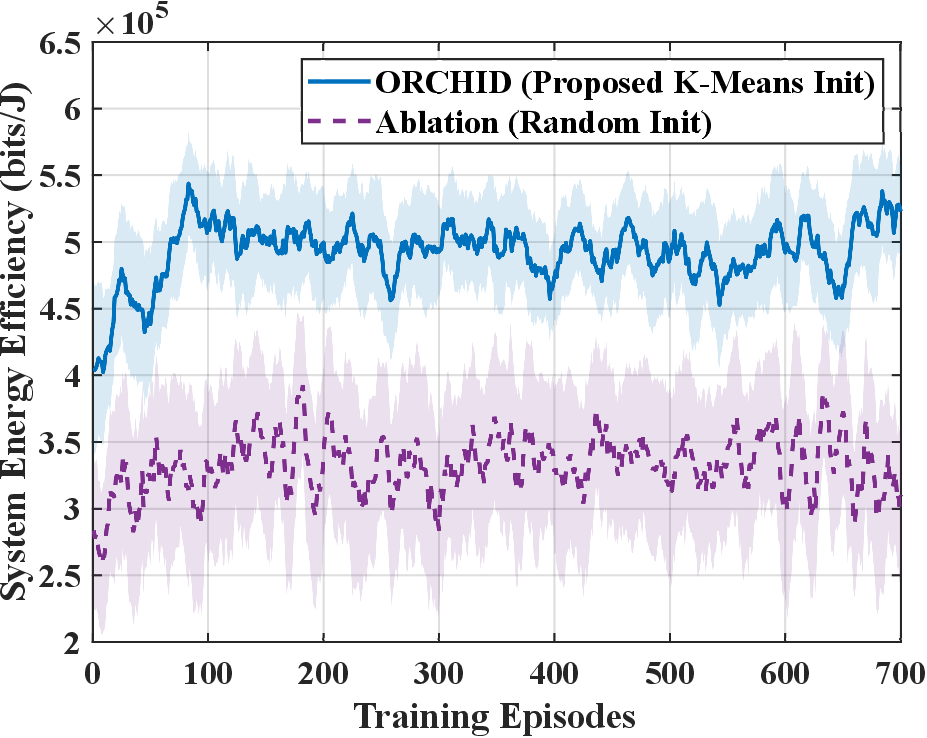}
    }

    \subfloat[Convergence of User Rate JFI.
    \label{fig:ablation_init:rate_jfi}]{
        \includegraphics[width=0.47\linewidth]{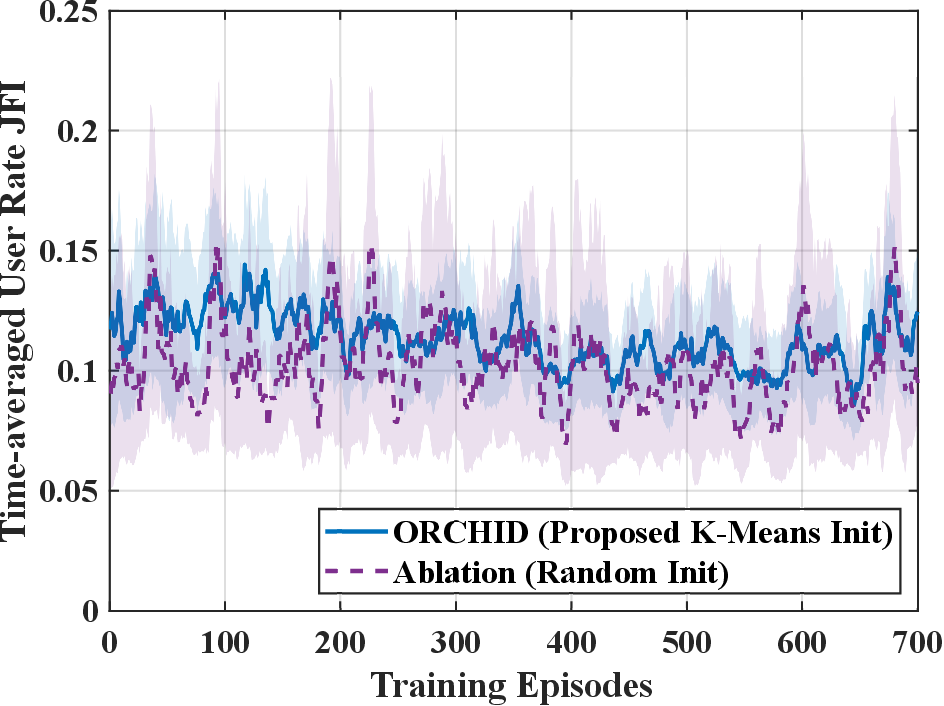}
    }
    \hfill
    \subfloat[Convergence of UAV Load JFI.
    \label{fig:ablation_init:load_jfi}]{
        \includegraphics[width=0.47\linewidth]{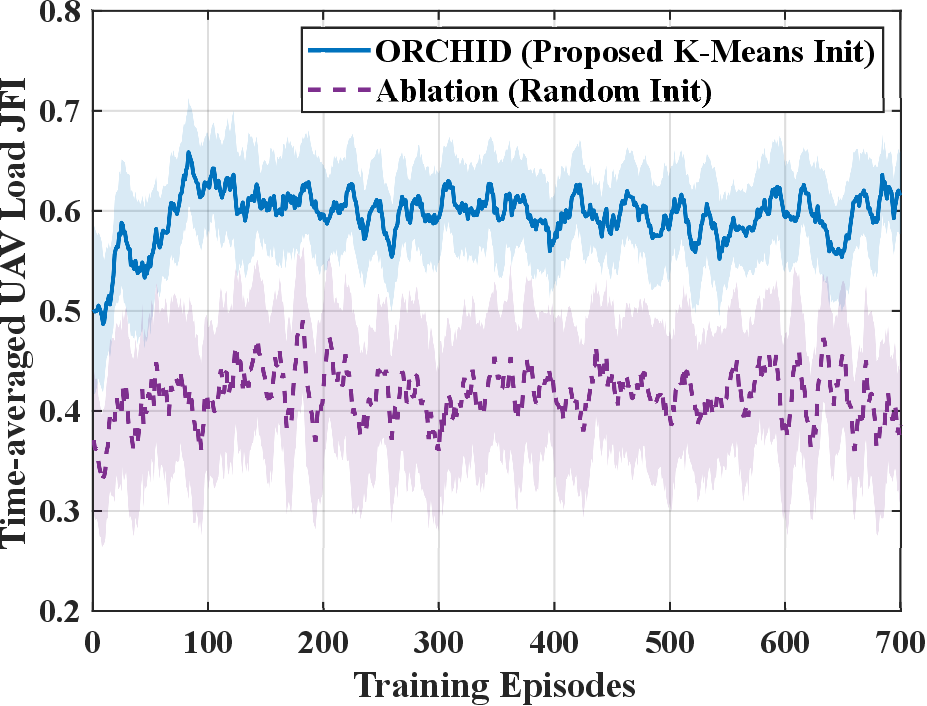}
    }
    \caption{Warm-start ablation over the standard training schedule. The curves show the mean over five independent training seeds, while the shaded regions denote one standard deviation. The GBS-aware K-Means initialization improves system EE and UAV-load fairness relative to random initialization.}
    \label{fig:ablation_init}
\end{figure}

\subsection{Strategy Analysis: MMF versus PF}
\label{sec:strategy_fairness}

To isolate the operational impact of bandwidth allocation during execution, trained policies ($E_{\max}=700$) under each scheme are evaluated over 1,000 paired inference episodes using synchronized random seeds. ORCHID-MMF applies the closed-form egalitarian allocation derived in Section~\ref{sec:theoretical_analysis}, whereas ORCHID-PF uses equal bandwidth sharing among users associated with the same serving node. We additionally report the terminal-step worst-case user rate, defined as the minimum instantaneous rate among associated users receiving positive-rate service at the final step of each episode.

Fig.~\ref{fig:strategy_fairness} presents the comparison for $N=6$. As shown in Fig.~\ref{fig:strategy_fairness:cdf_worst_rate}, MMF substantially improves bottleneck-rate protection. The median terminal-step worst-case rate increases from approximately $0.35\text{ Mbps}$ under PF to $0.58\text{ Mbps}$ under MMF, while the observed maximum increases from $0.85\text{ Mbps}$ to $2.17\text{ Mbps}$. This rightward shift is consistent with Proposition~\ref{prop:ebg_solution}, whereby users with lower spectral efficiency receive larger bandwidth fractions under the egalitarian allocation.

Fig.~\ref{fig:strategy_fairness:scatter_jfi_nee} further shows that the two strategies maintain comparable fleet-load fairness. For $N=6$, the centroid of ORCHID-PF is approximately $3.72\times10^5$~bits/J with an average load JFI of $0.396$, whereas the ORCHID-MMF centroid is approximately $4.08\times10^5$~bits/J with an average load JFI of $0.404$. Thus, the bottleneck-rate improvement obtained by MMF does not incur an evident loss of fleet-level load fairness or energy efficiency in this setting.

The same trend persists across $N\in\{4,6,8,10\}$ in Fig.~\ref{fig:scalability_analysis}. As shown in Fig.~\ref{fig:scalability:min_rate}, ORCHID-MMF maintains a mean terminal-step worst-case rate ranging from approximately $0.58$ to $0.61$~Mbps, consistently above the approximately $0.34$ to $0.36$~Mbps observed under ORCHID-PF. Meanwhile, Fig.~\ref{fig:scalability:load_jfi} shows comparable fleet-load fairness across the two strategies: MMF is slightly higher at smaller swarm sizes, whereas PF becomes marginally higher for $N\geq8$, reaching $0.378$ compared with $0.355$ for MMF at $N=10$. This crossover indicates that the user-level bottleneck protection provided by MMF does not necessarily translate into uniformly higher inter-UAV load fairness.

\begin{remark}[Performance Characteristics of MMF and PF]
The results highlight the distinction between user-level bottleneck protection and fleet-level load fairness. Under the considered fixed-SNR linear-in-bandwidth model, PF uses equal bandwidth sharing, whereas MMF allocates larger bandwidth fractions to users with lower spectral efficiency. Consequently, MMF provides consistently stronger worst-case-rate protection across the evaluated swarm sizes, while the two schemes maintain comparable UAV-load fairness, with their relative ordering varying slightly with swarm size.
\end{remark}

\begin{figure}[!t]
    \vspace{-.5em}
    \centering
    \subfloat[CDF of Worst-Case Rate.]{
        \includegraphics[width=0.47\linewidth]{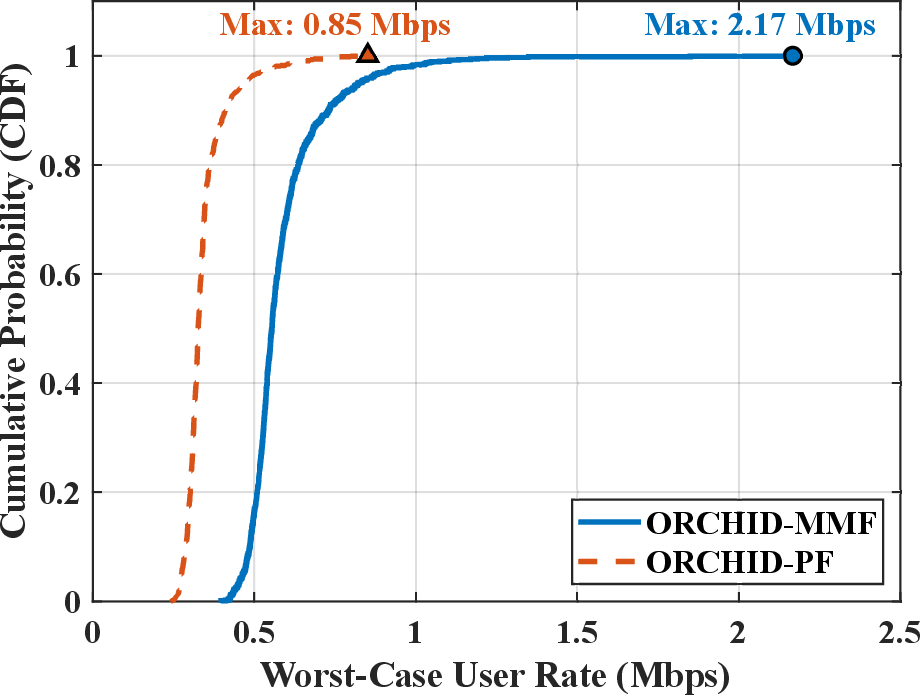}
        \label{fig:strategy_fairness:cdf_worst_rate}
    }\hfill
    \subfloat[EE and Load Fairness.]{
        \includegraphics[width=0.47\linewidth]{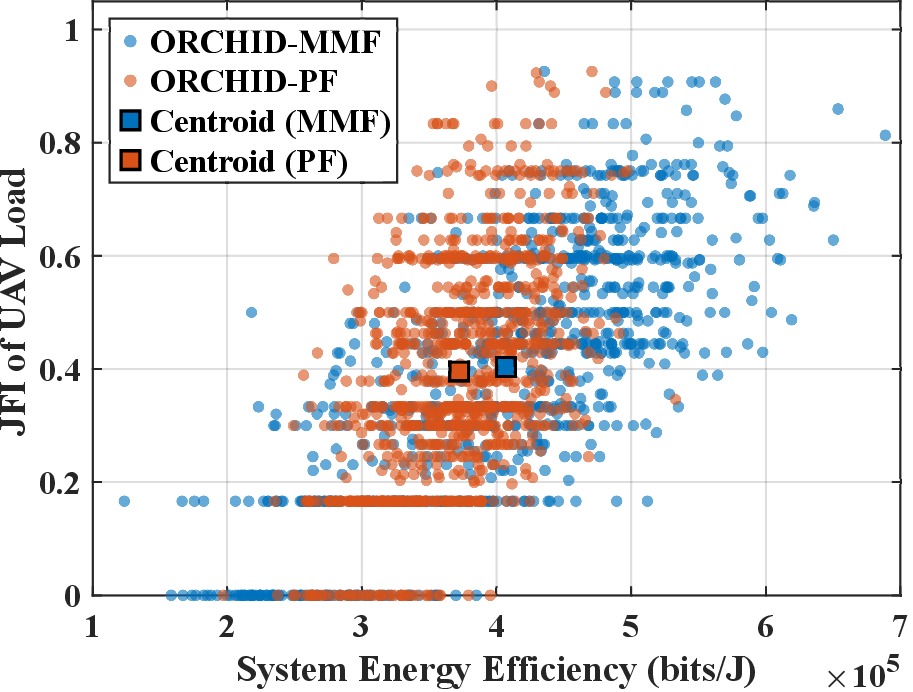}
        \label{fig:strategy_fairness:scatter_jfi_nee}
    }
    \caption{Strategy comparison for $N=6$: (a) empirical CDF of the terminal-step worst-case user rate and (b) joint distribution of system EE and fleet-load fairness.}
    \label{fig:strategy_fairness}
\end{figure}

\begin{figure}[!t]
    \vspace{-.5em}
    \centering
    \subfloat[Worst-Case Rate Scaling.]{
        \includegraphics[width=0.47\linewidth]{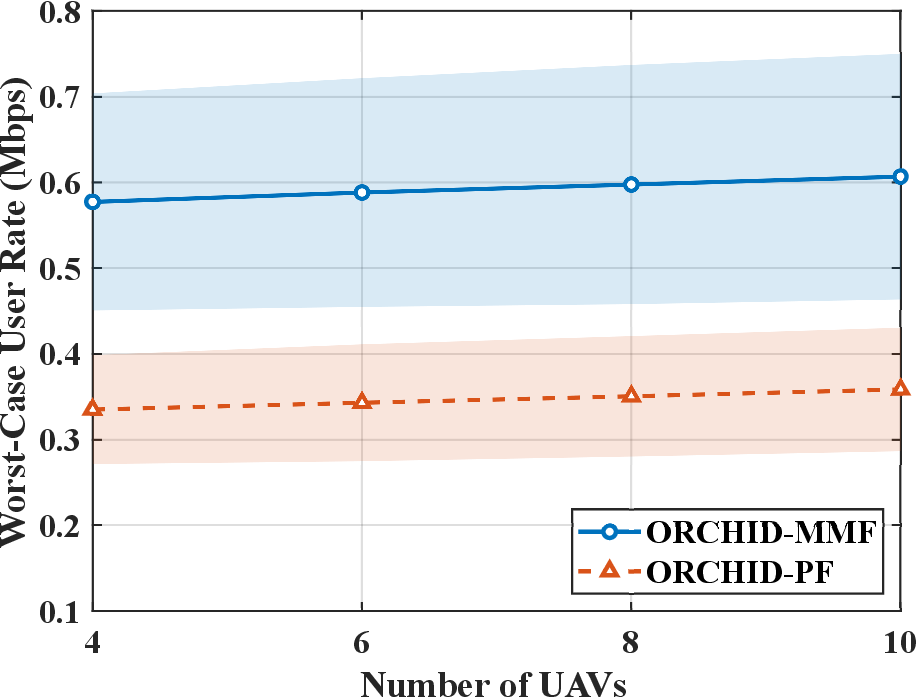}
        \label{fig:scalability:min_rate}
    }\hfill
    \subfloat[Fleet Load Fairness Scaling.]{
        \includegraphics[width=0.47\linewidth]{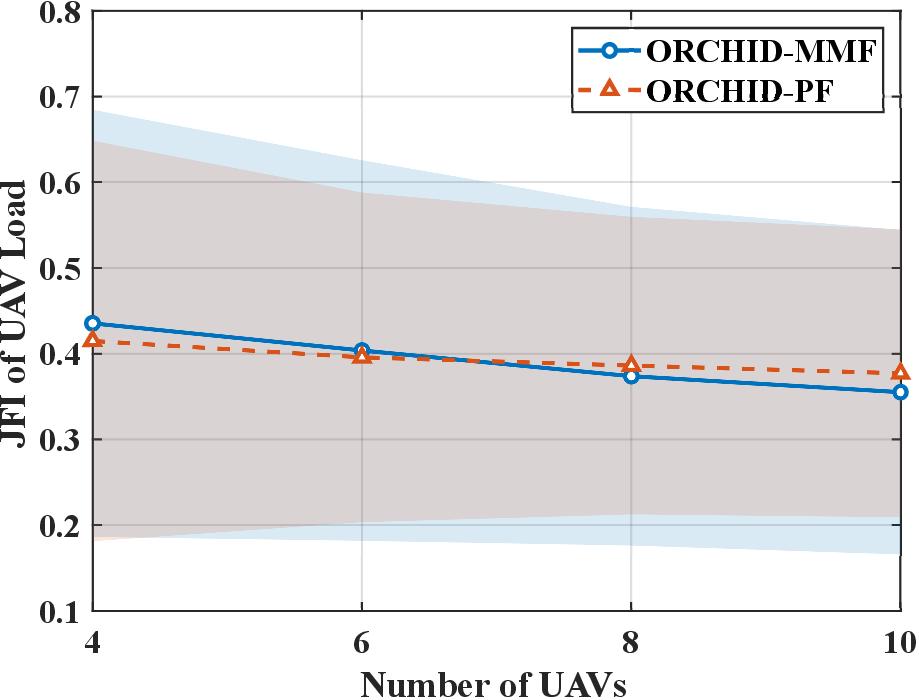}
        \label{fig:scalability:load_jfi}
    }
    \caption{Scalability comparison across $N\in\{4,6,8,10\}$: (a) mean terminal-step worst-case user rate and (b) fleet-load fairness.}
    \label{fig:scalability_analysis}
\end{figure}

\subsection{Summary and Discussion}
\label{sec:summary_discussion}

Overall, the experiments isolate the complementary roles of the ORCHID components. The warm-start ablation supports Phase~I as a topology-informed initialization that improves early energy efficiency and UAV-load balance, while the w/o-R\&F ablation supports R\&F as a conservative late-stage stabilization mechanism. The paired MMF/PF evaluation further shows that MMF provides substantially stronger served-user bottleneck-rate protection while maintaining comparable fleet-load fairness.

Across the zero-shot scalability evaluations, ORCHID preserves coverage and user-rate fairness comparable to the learned baselines, while exhibiting clearer advantages in system energy efficiency and UAV-load fairness. Under increasing user mobility, ORCHID retains the same energy- and load-balancing advantages and becomes the strongest learned mobile policy in user-rate fairness at the higher tested speeds. These results indicate that ORCHID is not designed to maximize every individual metric independently; rather, it seeks a robust operating balance among service availability, user-level fairness, fleet-level load distribution, and energy consumption under changing network conditions.

The experiments also clarify the roles of the individual mechanisms. The benefits of MMF are primarily reflected in bottleneck-user protection rather than universally higher aggregate fairness, whereas the gains of the learned orchestration policy are more evident in energy-efficient and load-balanced fleet coordination. Taken together, these observations complement the analytical EBG derivation and the qualified optimization interpretation of R\&F, providing empirical support under changes in user population, swarm size, and user mobility without implying universal optimality or formal convergence guarantees.

\section{Conclusion}
\label{sec:conclusion}

In this paper, we proposed ORCHID, a two-stage orchestration framework
that combines GBS-aware topology initialization with MAPPO-based
Reset-and-Finetune (R\&F) for mission-critical AGINs. The framework
integrates topology-aware warm start, fairness-oriented bandwidth
allocation, and event-triggered policy refinement to jointly improve
service fairness, energy efficiency, and training stability.

Simulation results show that ORCHID maintains competitive coverage and
user-rate fairness while achieving stronger system EE and UAV-load
fairness than the MADDPG baselines across the evaluated scalability and
mobility settings. The MMF configuration further provides substantially
stronger served-user bottleneck-rate protection than PF while maintaining
comparable fleet-load fairness. Future work will investigate
high-mobility scenarios, sensing uncertainty, dynamic spectrum sharing,
and more decentralized onboard inference.

\color{black}




\bibliographystyle{IEEEtran}
\bibliography{IEEEabrv,reference}

@STRING{IEEE_J_VT         = "{IEEE} Trans. Veh. Technol."}

@STRING{IEEE_J_COML       = "{IEEE} Commun. Lett."}

@STRING{IEEE_J_JSAC       = "{IEEE} J. Sel. Areas Commun."}

@STRING{IEEE_J_COM        = "{IEEE} Trans. Commun."}

@STRING{IEEE_J_WCOM       = "{IEEE} Trans. Wireless Commun."}

@STRING{IEEE_J_WCOML      = "{IEEE} Wireless Commun. Lett."}

@STRING{IEEE_J_GCN        = "{IEEE} Trans. Green Commun. Netw."}

@STRING{IEEE_J_IOT        = "{IEEE} Internet Things J."}

@STRING{IEEE_J_MC         = "{IEEE} Trans. Mobile Comput."}

@STRING{IEEE_O_ACC        = "{IEEE} Access"}

@STRING{IEEE_M_COM        = "{IEEE} Commun. Mag."}

@STRING{IEEE_O_CSTO       = "{IEEE} Commun. Surveys Tuts."}

@STRING{IEEE_M_NET        = "{IEEE} Netw."}

@article{9275613,
  author  = {Giordani, Marco and Zorzi, Michele},
  journal = IEEE_M_NET,
  title   = {Non-Terrestrial Networks in the {6G} Era: Challenges and Opportunities},
  year    = {2021},
  volume  = {35},
  number  = {2},
  pages   = {244-251},
  month   = {Mar./Apr.},
  doi     = {10.1109/MNET.011.2000493}
}

@article{10579820,
  author  = {Chen, Qian and Guo, Zheng and Meng, Weixiao and Han, Shuai and Li, Cheng and Quek, Tony Q. S.},
  journal = IEEE_O_CSTO,
  title   = {A Survey on Resource Management in Joint Communication and Computing-Embedded {SAGIN}},
  year    = {2025},
  volume  = {27},
  number  = {3},
  pages   = {1911-1954},
  month   = Jun,
  doi     = {10.1109/COMST.2024.3421523}
}

@article{11047530,
  author  = {Zhao, Wei and Cui, Shaoxin and Qiu, Wen and He, Zhiqiang and Liu, Zhi and Zheng, Xiao and Mao, Bomin and Kato, Nei},
  journal = IEEE_O_CSTO,
  title   = {A Survey on {DRL}-Based {UAV} Communications and Networking: {DRL} Fundamentals, Applications and Implementations},
  year    = {2026},
  volume  = {28},
  number  = {},
  pages   = {3911-3941},
  doi     = {10.1109/COMST.2025.3581912}
}

@article{7470933,
  author  = {Zeng, Yong and Zhang, Rui and Lim, Teng Joon},
  journal = IEEE_M_COM,
  title   = {Wireless communications with unmanned aerial vehicles: opportunities and challenges},
  year    = {2016},
  volume  = {54},
  number  = {5},
  pages   = {36-42},
  month   = may,
  doi     = {10.1109/MCOM.2016.7470933}
}

@article{8642333,
  author  = {Lai, Chuan-Chi and Chen, Chun-Ting and Wang, Li-Chun},
  journal = IEEE_J_WCOML,
  title   = {On-Demand Density-Aware {UAV} Base Station {3D} Placement for Arbitrarily Distributed Users With Guaranteed Data Rates},
  year    = {2019},
  volume  = {8},
  number  = {3},
  pages   = {913-916},
  month   = Jun,
  doi     = {10.1109/LWC.2019.2899599}
}

@article{9177297,
  author  = {Lai, Chuan-Chi and Wang, Li-Chun and Han, Zhu},
  journal = IEEE_J_MC,
  title   = {The Coverage Overlapping Problem of Serving Arbitrary Crowds in {3D} Drone Cellular Networks},
  year    = {2022},
  volume  = {21},
  number  = {3},
  month   = mar,
  doi     = {10.1109/TMC.2020.3019106}
}

@article{drones8060214,
  author         = {Wei, Dexing and Zhang, Lun and Liu, Quan and Chen, Hao and Huang, Jian},
  title          = {{UAV} Swarm Cooperative Dynamic Target Search: A {MAPPO}-Based Discrete Optimal Control Method},
  journal        = {Drones},
  volume         = {8},
  year           = {2024},
  number         = {6},
  article-number = {214},
  month          = may,
  doi            = {10.3390/drones8060214}
}

@inproceedings{lowe2017multi,
  author    = {Lowe, Ryan and Wu, Yi and Tamar, Aviv and Harb, Jean and Abbeel, Pieter and Mordatch, Igor},
  title     = {Multi-agent actor-critic for mixed cooperative-competitive environments},
  year      = {2017},
  booktitle = {International Conference on Neural Information Processing Systems (NeurIPS)},
  address   = {Long Beach, California, USA}
}

@inproceedings{NIPS2022_1787_Yu,
  author    = {Yu, Chao and Velu, Akash and Vinitsky, Eugene and Gao, Jiaxuan and Wang, Yu and Bayen, Alexandre and Wu, Yi},
  title     = {The surprising effectiveness of {PPO} in cooperative multi-agent games},
  year      = {2022},
  booktitle = {International Conference on Neural Information Processing Systems (NeurIPS)},
  address   = {New Orleans, LA, USA}
}

@article{9946428,
  author  = {Lai, Chuan-Chi and Bhola and Tsai, Ang-Hsun and Wang, Li-Chun},
  journal = IEEE_J_VT,
  title   = {Adaptive and Fair Deployment Approach to Balance Offload Traffic in Multi-{UAV} Cellular Networks},
  year    = {2023},
  volume  = {72},
  number  = {3},
  pages   = {3724-3738},
  month   = mar,
  doi     = {10.1109/TVT.2022.3221557}
}

@techreport{jain1984aquantitative,
  title       = {A Quantitative Measure of Fairness and Discrimination for Resource Allocation in Shared Systems},
  author      = {Jain, Raj and Chiu, Dah-Ming and Hawe, William},
  year        = {1984},
  month       = sep,
  institution = {Digital Equipment Corporation},
  note        = {{TR-301}}
}

@article{7486987,
  author  = {Mozaffari, Mohammad and Saad, Walid and Bennis, Mehdi and Debbah, M\acute{e}rouane},
  journal = IEEE_J_COML,
  title   = {Efficient Deployment of Multiple Unmanned Aerial Vehicles for Optimal Wireless Coverage},
  year    = {2016},
  volume  = {20},
  number  = {8},
  pages   = {1647-1650},
  month   = aug,
  doi     = {10.1109/LCOMM.2016.2578312}
}

@article{10122731,
  author  = {Mahmood, Asad and Vu, Thang Xuan and Chatzinotas, Symeon and Ottersten, Bj\ddot{o}rn},
  journal = IEEE_J_VT,
  title   = {Joint Optimization of {3D} Placement and Radio Resource Allocation for Per-{UAV} Sum Rate Maximization},
  year    = {2023},
  volume  = {72},
  number  = {10},
  pages   = {13094-13105},
  month   = oct,
  doi     = {10.1109/TVT.2023.3274815}
}

@inproceedings{SODA07_Arthur,
  author    = {Arthur, David and Vassilvitskii, Sergei},
  title     = {k-means++: the advantages of careful seeding},
  year      = {2007},
  isbn      = {9780898716245},
  booktitle = {The Eighteenth Annual ACM-SIAM Symposium on Discrete Algorithms},
  address   = {New Orleans, Louisiana}
}

@article{R2025102621,
  author  = {{Dhinesh Kumar R.} and Rammohan, A.},
  title   = {Optimizing {UAV} deployment for maximizing coverage and data rate efficiency using multi-agent deep deterministic policy gradient and {Bayesian} optimization},
  journal = {Physical Communication},
  volume  = {69},
  pages   = {102621},
  year    = {2025},
  month   = apr,
  doi     = {10.1016/j.phycom.2025.102621}
}

@article{11153951,
  author  = {Swain, Sipra and Ranjan Swain, Rakesh and Ranjan Senapati, Biswa and Mohan Khilar, Pabitra},
  journal = IEEE_O_ACC,
  title   = {An Optimized {2D} Ground Area Coverage Using {UAV}-Enabled Sensor Networks},
  year    = {2025},
  volume  = {13},
  number  = {},
  pages   = {161299-161310},
  month   = sep,
  doi     = {10.1109/ACCESS.2025.3608044}
}

@article{8618602,
  author  = {Zhang, Guangchi and Wu, Qingqing and Cui, Miao and Zhang, Rui},
  journal = IEEE_J_WCOM,
  title   = {Securing {UAV} Communications via Joint Trajectory and Power Control},
  year    = {2019},
  volume  = {18},
  number  = {2},
  pages   = {1376-1389},
  month   = feb,
  doi     = {10.1109/TWC.2019.2892461}
}

@article{zeng2019energy,
  author  = {Zeng, Yong and Xu, Jie and Zhang, Rui},
  journal = IEEE_J_WCOM,
  title   = {Energy Minimization for Wireless Communication With Rotary-Wing UAV},
  year    = {2019},
  volume  = {18},
  number  = {4},
  pages   = {2329-2345},
  month   = apr,
  doi     = {10.1109/TWC.2019.2902559}
}

@article{li2020UAV,
  author  = {Li, Xingwang and Wang, Qunshu and Liu, Yuanwei and Tsiftsis, Theodoros A. and Ding, Zhiguo and Nallanathan, Arumugam},
  journal = IEEE_J_WCOML,
  title   = {{UAV}-Aided Multi-Way {NOMA} Networks With Residual Hardware Impairments},
  year    = {2020},
  volume  = {9},
  number  = {9},
  pages   = {1538-1542},
  month   = sep,
  doi     = {10.1109/LWC.2020.2996782}
}

@article{10320337,
  author  = {Lei, Jiayi and Zhang, Tiankui and Mu, Xidong and Liu, Yuanwei},
  journal = IEEE_J_COM,
  title   = {{NOMA} for {STAR-RIS} Assisted {UAV} Networks},
  year    = {2024},
  volume  = {72},
  number  = {3},
  pages   = {1732-1745},
  month   = mar,
  doi     = {10.1109/TCOMM.2023.3333880}
}

@article{liu2020energy,
  author  = {Liu, Chi Harold and Ma, Xiaoxin and Gao, Xudong and Tang, Jian},
  journal = IEEE_J_MC,
  title   = {Distributed Energy-Efficient Multi-{UAV} Navigation for Long-Term Communication Coverage by Deep Reinforcement Learning},
  year    = {2020},
  volume  = {19},
  number  = {6},
  pages   = {1274-1285},
  month   = jun,
  doi     = {10.1109/TMC.2019.2908171}
}

@article{hu2020reinforcement,
  author  = {Hu, Jingzhi and Zhang, Hongliang and Song, Lingyang and Schober, Robert and Poor, H. Vincent},
  journal = IEEE_J_COM,
  title   = {Cooperative Internet of {UAV}s: Distributed Trajectory Design by Multi-Agent Deep Reinforcement Learning},
  year    = {2020},
  volume  = {68},
  number  = {11},
  pages   = {6807-6821},
  month   = nov,
  doi     = {10.1109/TCOMM.2020.3013599}
}

@article{ning2024multi,
  author  = {Ning, Zhaolong and Yang, Yuxuan and Wang, Xiaojie and Song, Qingyang and Guo, Lei and Jamalipour, Abbas},
  journal = IEEE_J_MC,
  title   = {Multi-Agent Deep Reinforcement Learning Based {UAV} Trajectory Optimization for Differentiated Services},
  year    = {2024},
  volume  = {23},
  number  = {5},
  pages   = {5818-5834},
  month   = may,
  doi     = {10.1109/TMC.2023.3312276}
}

@article{kang2023mappo,
  author  = {Kang, Hongyue and Chang, Xiaolin and Mišić, Jelena and Mišić, Vojislav B. and Fan, Junchao and Liu, Yating},
  journal = IEEE_J_IOT,
  title   = {Cooperative {UAV} Resource Allocation and Task Offloading in Hierarchical Aerial Computing Systems: A {MAPPO}-Based Approach},
  year    = {2023},
  volume  = {10},
  number  = {12},
  pages   = {10497-10509},
  month   = jun,
  doi     = {10.1109/JIOT.2023.3240173}
}

@article{9475471,
  author  = {Zhao, Mingxiong and Li, Wentao and Bao, Lingyan and Luo, Jia and He, Zhenli and Liu, Di},
  journal = IEEE_J_GCN,
  title   = {Fairness-Aware Task Scheduling and Resource Allocation in {UAV}-Enabled Mobile Edge Computing Networks},
  year    = {2021},
  volume  = {5},
  number  = {4},
  pages   = {2174-2187},
  month   = dec,
  doi     = {10.1109/TGCN.2021.3095070}
}

@inproceedings{8885992,
  author    = {Qin, Juan and Wei, Zhiqing and Qiu, Chen and Feng, Zhiyong},
  booktitle = {IEEE Wireless Communications and Networking Conference (WCNC)},
  title     = {Edge-Prior Placement Algorithm for {UAV}-Mounted Base Stations},
  year      = {2019},
  doi       = {10.1109/WCNC.2019.8885992},
  address   = {Marrakech, Morocco}
}

@article{10146333,
  author  = {Fu, Shu and Feng, Xue and Sultana, Ajmery and Zhao, Lian},
  journal = IEEE_J_WCOM,
  title   = {Joint Power Allocation and {3D} Deployment for {UAV-BSs}: A Game Theory Based Deep Reinforcement Learning Approach},
  year    = {2024},
  volume  = {23},
  number  = {1},
  pages   = {736-748},
  month   = feb,
  doi     = {10.1109/TWC.2023.3281812}
}

@article{al2014modeling,
  author  = {Al-Hourani, Akram and Kandeepan, Sithamparanathan and Lardner, Simon},
  journal = IEEE_J_WCOML,
  title   = {Optimal {LAP} Altitude for Maximum Coverage},
  year    = {2014},
  volume  = {3},
  number  = {6},
  pages   = {569-572},
  month   = dec,
  doi     = {10.1109/LWC.2014.2342736}
}

@book{chiu2013stochastic,
  title     = {Stochastic Geometry and its Applications},
  author    = {Chiu, Sung Nok and Stoyan, Dietrich and Kendall, Wilfrid S and Mecke, Joseph},
  year      = {2013},
  publisher = {John Wiley \& Sons}
}

@inproceedings{ash2020warm,
  author    = {Ash, Jordan T. and Adams, Ryan P.},
  title     = {On warm-starting neural network training},
  year      = {2020},
  booktitle = {International Conference on Neural Information Processing Systems (NeurIPS)},
  address   = {Vancouver, BC, Canada}
}

@inproceedings{nikishin2022primacy,
  title     = {The Primacy Bias in Deep Reinforcement Learning},
  author    = {Nikishin, Evgenii and Schwarzer, Max and D'Oro, Pierluca and Bacon, Pierre-Luc and Courville, Aaron},
  booktitle = {International Conference on Machine Learning (ICML)},
  address   = {Baltimore, Maryland, USA},
  year      = {2022}
}

@inproceedings{reddi2018convergence,
  author    = {Sashank J. Reddi and
               Satyen Kale and
               Sanjiv Kumar},
  title     = {On the Convergence of {Adam} and Beyond},
  booktitle = {International Conference on Learning Representations (ICLR)},
  year      = {2018},
  address   = {Vancouver, BC, Canada}
}

@article{zhan2025joint,
  author  = {Zhan, Cheng and Liu, Wei and Song, Kaifeng and Fan, Rongfei and Liu, Jun and Hu, Han},
  journal = IEEE_J_MC,
  title   = {Joint {UAV} Placement and Dependent Task Offloading in Multi-{UAV} {MEC} Networks: A Graph Attention Enhanced {DRL} Approach},
  year    = {2026},
  volume  = {25},
  number  = {4},
  pages   = {5285-5301},
  month   = apr,
  doi     = {10.1109/TMC.2025.3628608}
}

@article{jia2025distributionally,
  author  = {Jia, Ziye and Cui, Can and Dong, Chao and Wu, Qihui and Ling, Zhuang and Niyato, Dusit and Han, Zhu},
  journal = IEEE_J_MC,
  title   = {Distributionally Robust Optimization for Aerial Multi-Access Edge Computing via Cooperation of {UAVs} and {HAPs}},
  year    = {2025},
  volume  = {24},
  number  = {10},
  pages   = {10853-10867},
  month   = oct,
  doi     = {10.1109/TMC.2025.3571023}
}

@article{liu2022integrated,
  author  = {Liu, Fan and Cui, Yuanhao and Masouros, Christos and Xu, Jie and Han, Tony Xiao and Eldar, Yonina C. and Buzzi, Stefano},
  journal = IEEE_J_JSAC,
  title   = {Integrated Sensing and Communications: Toward Dual-Functional Wireless Networks for {6G} and Beyond},
  year    = {2022},
  volume  = {40},
  number  = {6},
  pages   = {1728-1767},
  month   = jun,
  doi     = {10.1109/JSAC.2022.3156632}
}

@article{cui2021integrating,
  author  = {Cui, Yuanhao and Liu, Fan and Jing, Xiaojun and Mu, Junsheng},
  journal = IEEE_M_NET,
  title   = {Integrating Sensing and Communications for Ubiquitous {IoT}: Applications, Trends, and Challenges},
  year    = {2021},
  volume  = {35},
  number  = {5},
  pages   = {158-167},
  month   = {Sep./Oct.},
  doi     = {10.1109/MNET.010.2100152}
}

@article{tang2025mulvr,
  author  = {Tang, Xiao-Wei and Huang, Yi and Shi, Yunmei and Wu, Qingqing},
  journal = IEEE_J_WCOM,
  title   = {{MUL-VR}: Multi-{UAV} Collaborative Layered Visual Perception and Transmission for Virtual Reality},
  year    = {2025},
  volume  = {24},
  number  = {4},
  pages   = {2734-2749},
  month   = apr,
  doi     = {10.1109/TWC.2024.3524275}
}

@article{nash1950bargaining,
  title   = {The bargaining problem},
  author  = {Nash, John F.},
  journal = {Econometrica},
  pages   = {155--162},
  volume  = {18},
  number  = {2},
  month   = apr,
  year    = {1950}
}

@article{kalai1977proportional,
  title   = {Proportional solutions to bargaining situations: interpersonal utility comparisons},
  author  = {Kalai, Ehud},
  journal = {Econometrica},
  volume  = {45},
  number  = {7},
  pages   = {1623--1630},
  month   = oct,
  year    = {1977}
}

@article{akdeniz2014millimeter,
  author  = {Akdeniz, Mustafa Riza and Liu, Yuanpeng and Samimi, Mathew K. and Sun, Shu and Rangan, Sundeep and Rappaport, Theodore S. and Erkip, Elza},
  journal = IEEE_J_JSAC,
  title   = {Millimeter Wave Channel Modeling and Cellular Capacity Evaluation},
  year    = {2014},
  volume  = {32},
  number  = {6},
  month   = jun,
  doi     = {10.1109/JSAC.2014.2328154}
}

@techreport{3gpp_tr36931_v17,
  author      = {{3GPP}},
  title       = {{Evolved Universal Terrestrial Radio Access (E-UTRA);
                 Radio Frequency (RF) requirements for LTE Pico Node B}},
  institution = {3rd Generation Partnership Project (3GPP)},
  type        = {Technical Report},
  number      = {TR 36.931 V17.0.0},
  year        = {2022},
  month       = apr
}
\ifCLASSOPTIONcaptionsoff  \newpage \fi 

\begin{IEEEbiography}[{\includegraphics[width=1in,height=1.25in,clip,keepaspectratio]{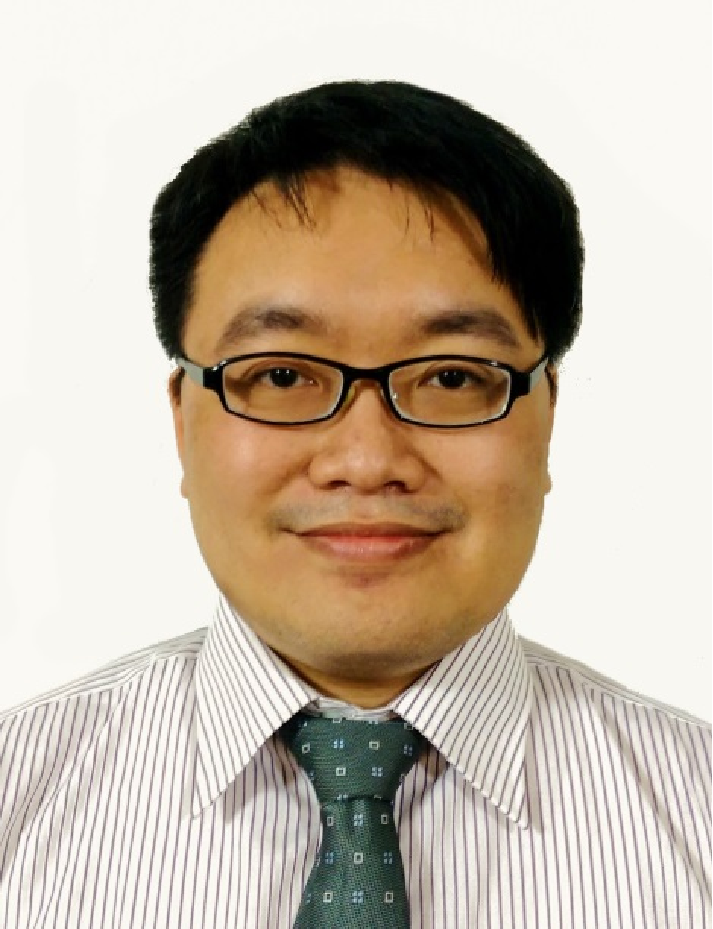}}]{Chuan-Chi Lai}
    (Member, IEEE) received the Ph.D. degree in Computer Science and Information Engineering from the National Taipei University of Technology, Taiwan, in 2017. He held research and faculty positions at National Chiao Tung University and Feng Chia University prior to his current role. Since 2024, he has been an Assistant Professor with the Department of Communications Engineering, National Chung Cheng University, Chiayi, Taiwan. His research interests include mobile edge computing, UAV networks, and AI for wireless communications. Dr. Lai was a recipient of the Postdoctoral Researcher Academic Research Award from the NSTC, Taiwan, in 2019, and Best Paper Awards at WOCC (2018, 2021) and ICUFN (2015).
\end{IEEEbiography}

\begin{IEEEbiography}[{\includegraphics[width=1in,height=1.25in,clip,keepaspectratio]{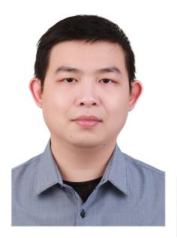}}]{Chi Jai Choy}
    received the B.S. degree in Computer Science and Information Engineering from Feng Chia University, Taichung, Taiwan, in 2021, and the M.S. degree in Computer Science and Information Engineering from Feng Chia University, Taichung, Taiwan, in 2025. His research interests include multi-agent reinforcement learning, UAV-enabled communications, and energy-efficient wireless networks.
\end{IEEEbiography}

\end{document}